\documentclass[10pt,a4paper]{article}
\usepackage[utf8]{inputenc}
\usepackage[english]{babel}
\usepackage{amsmath}
\usepackage{amsfonts}
\usepackage{amssymb}
\usepackage{amsthm}
\usepackage{cite}
\usepackage{graphicx}
\usepackage[left=2cm,right=2cm,top=2cm,bottom=2cm]{geometry}
\usepackage{authblk}
\usepackage[colorlinks,citecolor=red,urlcolor=blue,bookmarks=false,hypertexnames=true]{hyperref} 
\newtheorem{te}{Theorem}
\newtheorem{coro}{Corollary}
\newtheorem{lem}{Lemma}
\newtheorem{prop}{Proposition}

\newtheorem{de}{Definition}
\newtheorem{ex}{Example}

\providecommand{\keywords}[1]
{\small	\textbf{\textit{Keywords:}} #1}
\title{Scaling symmetries and canonoid transformations in Hamiltonian systems}
\author[1]{R. Azuaje}
\author[2]{A. Bravetti}
\affil[1]{Departamento de F\'isica, Universidad Aut\'onoma Metropolitana Unidad Iztapalapa,
San Rafael Atlixco 186, 09340 Cd. Mx., M\'exico}
\affil[2]{Instituto de Investigaciones en Matem\'aticas Aplicadas y en Sistemas, 
Universidad Nacional Aut\'onoma de M\'exico, A.~P.~70543, M\'exico, DF 04510, M\'exico}
\begin{document}
\maketitle

\begin{abstract}
We investigate various types of symmetries and their mutual relationships in Hamiltonian systems defined on manifolds with different geometric structures:
symplectic, cosymplectic, contact and cocontact.
In each case we pay special attention to non-standard (non-canonical) symmetries,
in particular scaling symmetries and canonoid transformations, as they provide new interesting tools for the qualitative study of these systems.   
Our main results are the characterizations of these non-standard symmetries and the analysis of their relation with conserved (or dissipated) quantities.
\end{abstract}


\section{Introduction}
\label{sec1}

Symmetry methods are arguably among the most important and most venerable tools
in the study of physical systems~\cite{olver1993applications}. One of the
consequences of the existence of symmetries is their association with conserved quantities given by the celebrated theorem by Noether. One can derive Noether's theorem either in its original context, the variational formulation of mechanics and field theory, or using the corresponding geometric re-formulation~\cite{Arnold78,KS2011}. 
In the standard geometric formulation of mechanics, one usually has a symplectic manifold and a Hamiltonian dynamics defined on such manifold. 
In this context one can define \emph{Noether symmetries} as symplectic (locally-Hamiltonian) vector fields that preserve the Hamiltonian function defining the dynamics.
Then Noether's result is given automatically (at least locally) by the anti-symmetry of the symplectic form~\cite{Lee2012,AMRC2019,Roman2020,AKN2006}.

There are several generalizations of Noether's theorem. 
On the one hand, one can consider a symplectic Hamiltonian system admitting an action of a symmetry group that is Hamiltonian and derive a more general version of Noether's theorem, one that includes ``several conserved quantities'' in one map, namely, the momentum map~\cite{marsden2013introduction,libermann2012symplectic}.
On the other hand, one can consider different types of Hamiltonian systems. 
For instance, time-dependent Hamiltonian systems
are best described geometrically in the context of cosymplectic geometry~\cite{LR89,CLL92,LS2017}; some dissipative and thermostatted mechanical systems are more naturally incorporated within the context of contact geometry~\cite{LL2019,LS2017,BCT2017,Bravetti2017}. More recently, cocontact
Hamiltonian systems have been studied, which generalize both cosymplectic and contact systems at once and allow for a nice geometric analysis of time-dependent contact Hamiltonian 
dynamics~\cite{Letal2022,rivas2022lagrangian,Azuaje2023}.
In all these contexts one can think of appropriate versions 
of Noether's theorem~\cite{Lee2012,Torres2020,Jovanovic2016,Azuaje2022,BG2021,GGMRR2020,LL2020}.
Finally, one may also generalize the concept of standard Noether symmetries and the related theorem  
(see e.g.~\cite{sarlet1981generalizations, GGMRR2020,BG2021}).
In particular, scaling symmetries and canonoid transformations are two natural 
and most important generalizations of standard Noether symmetries and canonical transformations respectively.
Scaling symmetries have long played a special role in the analysis of mathematical and physical phenomena~\cite{poincare2003science,Chenciner1997alinfini,gryb2021scale}. 
Moreover, a generalized Noether theorem for scaling symmetries has been found recently in the symplectic case~\cite{zhang2020generalized} and it has been argued in~\cite{BG2021} that the contact context may be more appropriate in order to fully understand the extent of the theorem. 
Remarkably, scaling symmetries are also central to obtain the intrinsic description of the universe given in the theory of ``shape dynamics''~\cite{Sloan2018,Sloan2021,gryb2021scale}. For all these reasons there has been a surge of interest to study them from various perspectives (see also~\cite{BJS2022}).
Canonoid transformations are an important generalization of standard canonical transformations, as they also preserve the Hamiltonian nature of the system. They have been defined first in~\cite{currie1972canonical}. Although they are far less famous than their canonical counterparts, they have been studied thoroughly in~\cite{NOT87,TNO89,CR88,CR89,CFR2013,TC2018,TCA2022} for the symplectic case,
in~\cite{RS2015} for the Poisson case, and more recently in~\cite{AE2023} for the cosymplectic, contact and cocontact cases.
 
In this work, we consider several possible generalizations of standard symplectic Hamiltonian systems, namely, cosymplectic, contact and cocontact systems, and define both scaling symmetries and canonoid transformations on each category. Then we provide some useful characterizations of these symmetries, 
we study their mutual relationships and in each case we derive Noether-like theorems that associate invariant or dissipated quantities to these symmetries.

To make the paper as clear and self-contained as possible, 
we start in Section~\ref{sec2} with a review of the known results and definitions
in the symplectic case. 
Then in the ensuing sections we follow the scheme detailed in Section~\ref{sec2} but we extend all the definitions and results to the cosymplectic (Section~\ref{sec3}), contact (Section~\ref{sec4}), and cocontact (Section~\ref{sec5}) case respectively. 
All the relevant definitions for these classes of systems are given in Appendix~\ref{sec:appA}.
Finally, in Section~\ref{sec:conclusions} we present the conclusions and outline some possible directions for future work. Throughout the paper we use Einstein's summation convention (i.e., a summation over repeated indices is assumed).

\section{Symmetries for symplectic Hamiltonian systems}
\label{sec2}

In this section we review some known definitions and results for symplectic Hamiltonian systems in order to set the stage and the notation for the results that will be presented in the next sections (we refer to Appendix~\ref{sec:appA} when needed for further more basic definitions and properties).

Let $(M,\omega,H)$ be a symplectic Hamiltonian system (see Appendix~\ref{sec:appA}).
\begin{de}
A smooth vector field $V$ on $M$ is called \emph{an infinitesimal symmetry of $(M,\omega,H)$} if $L_{V}\omega=L_{V}H=0$.
\end{de}
Note that this type of symmetries could be also naturally called \emph{Noether symmetries}, as they are the symmetries normally 
involved in the basic version of Noether's theorem.
Indeed, Noether's theorem reads (see~\cite{Lee2012,BG2021} for the proof and~\cite{RS2015,MSZ2020} for analogous results on Poisson manifolds):
\begin{te}
\begin{enumerate}
\item[i)] If $f$ is a constant of motion of $(M,\omega,H)$, then its associated Hamiltonian vector field 
$X_{f}$ is an infinitesimal symmetry of $(M,\omega,H)$. 
\end{enumerate}
Reciprocally,
\begin{enumerate}
\item[ii)] if $V$ is an infinitesimal symmetry of $(M,\omega,H)$, then there is a (possibly-local) function $f$ such that $V=X_{f}$ and $f$ 
is a constant of motion.
\end{enumerate}
\end{te}

Following the language employed in~\cite{Prince83,CFR2013,GGMRR2020}, we may consider the following more general definition of symmetry of a dynamics:
\begin{de}\label{def:general_infdynsym}
A smooth vector field $W$ on $M$ is called \emph{an infinitesimal dynamical symmetry of a dynamics $X$} if $[W,X]=0$. 
\end{de}
Clearly, when we restrict the above definition to the case in which the dynamics is given by a symplectic Hamiltonian system, we obtain
\begin{de}
A smooth vector field $W$ on $M$ is called \emph{an infinitesimal dynamical symmetry of $(M,\omega,H)$} if $[W,X_{H}]=0$. 
\end{de}
From the fact that the assignment $f\mapsto X_{f}$ is a Lie algebra 
antihomomorphism between the Lie algebras $(C^{\infty}(M),\lbrace,\rbrace)$ 
and $(\mathfrak{X}(M),[,])$, i.e., $X_{\lbrace f,g\rbrace}=-[X_{f},X_{g}]$, (see Appendix~\ref{sec:appA}), we have the following result: 
\begin{prop}
Every infinitesimal symmetry is an infinitesimal dynamical symmetry.
\end{prop}
\begin{proof}
Let $V$ be an infinitesimal symmetry, then $V=X_{f}$ for some constant of motion $f$. We have
\begin{equation}
[X_{f},X_{H}]=X_{\lbrace H,f\rbrace}=0
\end{equation}
since $f$ is a constant of motion. So we conclude that $V$ is an infinitesimal dynamical symmetry.
\end{proof}
The converse statement is not true. 
Indeed, if $X\in\mathfrak{X}(M)$ is such that $L_{X}H=\lambda_{1}H$ and 
$L_{X}\omega=\lambda_{2}\omega$ for some $\lambda_{1},\lambda_{2}\in\mathbb{R}$, 
then $[X,X_{H}]=(\lambda_{1}-\lambda_{2})X_{H}$ (cf.~Proposition 1 in~\cite{BG2021}, but beware that in this reference infinitesimal symmetries are called Noether symmetries); 
so for the case $\lambda_{1}=\lambda_{2}\neq 0$ we have that $X$ is an infinitesimal dynamical symmetry but not an infinitesimal symmetry. 
This is the case of scaling symmetries of degree 1, as we are now going to prove. 
We start by providing the general definition of a scaling symmetry in the symplectic case, which reads (see also~\cite{BJS2022}): 
\begin{de}
$X\in\mathfrak{X}(M)$ is a \emph{scaling symmetry of degree $\Lambda\in\mathbb{R}$}  for the Hamiltonian system $(M,\omega,H)$ if
\begin{enumerate}
\item[i)] $L_{X}\omega=\omega$ 
and
\item[ii)] $L_{X}H=\Lambda H$.
\end{enumerate}
\end{de}
From the first condition in their definition, scaling symmetries are not Hamiltonian vector fields (this comes from the fact that $L_{X_{f}}\omega=0$ for any function $f\in C^{\infty}(M)$)
and thus they are not infinitesimal symmetries.
Moreover, if $X$ is a scaling symmetry of degree $\Lambda$, 
then $[X,X_{H}]=(\Lambda-1)X_{H}$. Therefore we conclude that 
scaling symmetries of degree~1 are infinitesimal dynamical symmetries which are not infinitesimal symmetries.

It is well known that infinitesimal symmetries are related to one-parameter groups of canonical transformations. 
To recall this, we start with a geometrical definition of canonical transformations~\cite{AMRC2019,AE2023}.
\begin{de}
\emph{A canonical transformation} for the symplectic manifold $(M,\omega)$ is a possibly-local diffeomorphism $F$ on $M$ such that $F^{*}\omega=\omega$.
\end{de}
Given a diffeomorphism $F:M\longrightarrow M$,  we know that $F^{*}\omega$ is a symplectic structure on $M$, so around any point $p\in M$ there are local coordinates $(Q^{1},\cdots,Q^{n},P_{1},\cdots,P_{n})$ such that
\begin{equation}
F^{*}\omega=dQ^{i}\wedge dP_{i}.
\end{equation} 
If we think of $F$ locally as a coordinate transformation $(q^{1},\cdots,q^{n},p_{1},\cdots,p_{n})\longmapsto (Q^{1},\cdots,Q^{n},P_{1},\cdots,P_{n})$, 
then $F$ is a canonical transformation if and only if
\begin{equation}\label{eqcpb}
\lbrace Q^{i},Q^{j}\rbrace=\lbrace P_{i},P_{j}\rbrace=0 \hspace{1cm}\textit{and}\hspace{1cm}\lbrace Q^{i},P_{j}\rbrace=\delta^{i}_{j}.
\end{equation}
Let $V$ be a smooth vector field on $M$ and $\varphi$ its flow (for a brief review of flows on smooth manifolds see \cite{Lee2012}). 
By definition, $V$ is an  infinitesimal symmetry of $(M,\omega,H)$ if and only if $\omega$ and $H$ are invariant under the flow $\varphi$ of $V$ ($\varphi_{s}^{*}\omega=\omega$ and $\varphi_{s}^{*}H=H$), so we have the following proposition.
\begin{prop}\label{prop:symm-canonical-sympl}
If $V$ is an infinitesimal symmetry of $(M,\omega,H)$, 
then the flow 
of $V$ defines a one-parameter group of canonical transformations
that leave the Hamiltonian invariant; 
reciprocally, if we have a one-parameter group of canonical transformations that leave the Hamiltonian $H$ invariant, then the infinitesimal generator of the flow defined by such one-parameter group is an infinitesimal symmetry of $(M,\omega,H)$.
\end{prop}

Similarly to the above result, in~\cite{CFR2013} one-parameter groups of canonoid transformations are studied by using their infinitesimal generators. 
Classically a canonoid transformation is  defined as a coordinate transformation preserving the Hamiltonian form of the equations of motion~\cite{SC72,NOT87,TCA2022}. For our purposes it is more convenient to follow the geometric definition presented in~\cite{CR88,AE2023,CFR2013}.
\begin{de}
We say that a possibly-local diffeomorphism $F:M\longrightarrow M$ is \emph{a canonoid transformation} for the Hamiltonian system $(M,\omega,H)$ if there exists a function $K\in C^{\infty}(M)$ such that 
\begin{equation}
X_{H}\lrcorner F^{*}\omega=dK.
\end{equation}
\end{de}
In general, canonoid transformations do not preserve the Poisson bracket, so that if we think of $F$ locally as a coordinate transformation $$(q^{1},\cdots,q^{n},p_{1},\cdots,p_{n})\longmapsto (Q^{1},\cdots,Q^{n},P_{1},\cdots,P_{n}),$$ 
then the new coordinates do not necessarily satisfy (\ref{eqcpb}). 
Nevertheless, the equations of motion in the coordinates $(Q^{1},\cdots,Q^{n},P_{1},\cdots,P_{n})$ read
\begin{equation}
\dot{Q^{i}} = \frac{\partial K}{\partial P_{i}},\qquad
\dot{P_{i}} =-\frac{\partial K}{\partial Q^{i}}\,.
\end{equation}
Clearly, every canonical transformation is also (trivially) a canonoid transformation. 
However, the converse statement is not true. Indeed, we can find examples of canonoid transformations that are not canonical 
in Refs.~\cite{NOT87,CR88}.

Now let us focus on the relationship between canonoid transformations, their infinitesimal generators and the associated constants of motion.
Let $X$ be a smooth vector field on $M$ and $\varphi$ its flow. 
$X$ is the infinitesimal generator of a one-parameter group of canonoid transformations if and only if there exists a family $\lbrace K_{s}\rbrace$ of functions on $M$ such that
\begin{equation}
X_{H}\lrcorner \varphi_{s}^{*}\omega=dK_{s},
\end{equation}
which is equivalent to $d(X_{H}\lrcorner \varphi_{s}^{*}\omega)=0$, so $X$ is the infinitesimal generator of a one-parameter group of canonoid transformations if and only if 
\begin{equation}
\label{eqinfcanonoid}
L_{X_{H}}\varphi_{s}^{*}\omega=0.
\end{equation}
By taking $\frac{d}{ds}|_{s=0}$ in equation (\ref{eqinfcanonoid}), 
we have that $X$ is the infinitesimal generator of a one-parameter group of canonoid transformations if and only if
\begin{equation}
L_{X_{H}}L_{X}\omega=0,
\end{equation}
which is equivalent to the fact that there exists (at least locally) a function $K$ on $M$ such that
\begin{equation}
\label{eqcanfunc}
X_{H}\lrcorner L_{X}\omega=dK. 
\end{equation}
We have that the function $K$ in equation (\ref{eqcanfunc}) is a constant of motion~\cite{CFR2013}. 
Moreover, the following theorem provides a useful characterization of the infinitesimal generators of 
one-parameter groups of canonoid transformations
(see~\cite{CFR2013} for the proof and~\cite{RS2015} for an analogous statement in the Poisson case).
\begin{te}
\label{teinfcanonoid}
$X$ is the infinitesimal generator of a one-parameter group of canonoid transformations for $(M,\omega,H)$ if and only if $[X,X_{H}]$ is a Hamiltonian vector field. In this case its Hamiltonian function is $L_{X}H-K$, where $K$ is given in equation~\eqref{eqcanfunc}.
\end{te}

In the remaining of this section we comment on two direct consequences of Theorem~\ref{teinfcanonoid} that are of special interest for us.
\begin{coro}
Every scaling symmetry is the infinitesimal generator of a one-parameter group of canonoid transformations.    
\end{coro}
\begin{proof}
If $X$ is a scaling symmetry for $(M,\omega,H)$, then $[X,X_{H}]$ is a Hamiltonian vector field. 
Indeed, if $X$ is a scaling symmetry, then $[X,X_{H}]=(\Lambda-1)X_{H}$ (where $\Lambda$ is the degree of $X$) is the Hamiltonian vector field with Hamiltonian function $(\Lambda-1)H$. 
We can further check that in this case $K=H$, since, by equation (\ref{eqcanfunc}), we must have $(\Lambda-1)H=\Lambda H-H=L_{X}H-K$.
\end{proof}

Infinitesimal dynamical symmetries are not necessarily scaling symmetries. Nevertheless, we can directly prove an analogous statement.
\begin{coro}
Every infinitesimal dynamical symmetry $X$ is the infinitesimal generator of a one-parameter group of canonoid transformations.
\end{coro}
\begin{proof}
If $X$ is an infinitesimal dynamical symmetry of $(M,\omega,H)$, then $[X,X_{H}]=0$ is a Hamiltonian vector field with constant Hamiltonian function $c$. 
Therefore, according to Theorem~\ref{teinfcanonoid}, $X$ is the generator of a one-parameter group of canonoid transformations.
\end{proof}

In addition, we have that for each $s$ the new Hamiltonian function is $K_{s}=\varphi_{s}^{*}H$. Indeed, we know that $X_{H}$ is invariant under the flow $\varphi$ of $X$ ($\varphi_{s \ast}X_{H}=X_{H}$) and for every $Y\in\mathfrak{X}(M)$ we have
\begin{equation}
\begin{split}
(X_{H}\lrcorner \varphi_{s}^{*}\omega)(Y) &=(\varphi_{s}^{*}\omega)(X_{H},Y)\\
&=\omega(\varphi_{s \ast}X_{H},\varphi_{s \ast}Y)\\
&=(X_{H}\lrcorner\omega)(\varphi_{s \ast}Y)\\
&=\varphi_{s}^{*}(X_{H}\lrcorner\omega)(Y)\\
&=(\varphi_{s}^{*}dH)(Y)\\
&=(d\varphi_{s}^{*}H)(Y)
\end{split}
\end{equation}
where $\varphi_{s \ast}Y$ is the pushforward of $Y$ by $\varphi_{s}$, i.e., the vector field $X_{H}$ is Hamiltonian with respect to the symplectic structure $\varphi_{s}^{*}\omega$ with Hamiltonian function $K_{s}=\varphi_{s}^{*}H$.

\section{Symmetries for cosymplectic Hamiltonian systems}
\label{sec3}

In this section we review some recent results concerning canonoid transformations in the cosymplectic setting and introduce new results regarding scaling symmetries in this geometry (we refer to Appendix~\ref{sec:appA} for more definitions and properties).

Let $(M,\Omega,\eta,H)$ be a cosymplectic Hamiltonian system. Following the definition of infinitesimal symmetry in the symplectic framework, we say that:
\begin{de}
A smooth vector field $V$ on $M$ is \emph{an infinitesimal symmetry} of $(M,\Omega,\eta,H)$ 
if $L_{V}\Omega=L_{V}H=0$ and $V\lrcorner\eta=0$ (this implies that $L_{V}\eta=0$).
\end{de}
In~\cite{Torres2020,Jovanovic2016,Azuaje2022} one can find a version of Noether's theorem for the special case where $M=\mathbb{R}\times T^*Q$,
$\Omega=dp_{i}\wedge dq^{i}-dH\wedge dt$ and $\eta=dt$. In such case one has that
any constant of motion (which possibly depends explicitly on time) is associated to a symmetry $X$ of $\Omega$ ($L_{X}\Omega=0$), 
without needing to restrict to the case $L_{X}H=0$. 

It is easy to prove the following (new) version of Noether's theorem for Hamiltonian systems on general cosymplectic manifolds: 
\begin{te}
\begin{enumerate}
\item[i)]
If $f$ is a constant of motion of $(M,\Omega,\eta,H)$ such that $Rf=0$ 
($f$ does not depend explicitly on time), 
then $X_{f}$ is an infinitesimal symmetry of $(M,\Omega,\eta,H)$. 
\end{enumerate}
Reciprocally,
\begin{enumerate}
\item[ii)] if $V$ is an infinitesimal symmetry of $(M,\Omega,\eta,H)$, 
then there is a (possibly-local) function $f$ such that $V=X_{f}$ and $f$ is a constant of motion with $Rf=0$.
\end{enumerate}
\end{te}
\begin{proof}
\begin{enumerate}
\item[i)] We have that $L_{X_{f}}\Omega=X_{f}\lrcorner d\Omega+d(X_{f}\lrcorner \Omega)=
d(df-(Rf)\eta)=-d(Rf)\wedge\eta=0$ and $X_{f}\lrcorner\eta=0$. In addition 
\begin{equation}
L_{X_{f}}H=X_{f}H=\lbrace H,f\rbrace=R(f)=0.
\end{equation}
Therefore $X_{f}$ is an infinitesimal symmetry of $(M,\Omega,\eta,H)$.
\item[ii)] 
In this case we have
\begin{equation*}
\begin{split}
L_{V}\Omega=0 &\Longrightarrow V\lrcorner d\Omega + d(V\lrcorner\Omega)=0\\
&\Longrightarrow d(V\lrcorner\Omega)=0\\
&\Longrightarrow V\lrcorner\Omega=df\,.
\end{split}
\end{equation*}
On the other hand
\begin{equation}
Rf=df(R)=\Omega(V,R)=-(R\lrcorner\Omega)(V)=0\,.
\end{equation}
Finally,
\begin{equation}
\begin{split}
&L_{E_{H}}f = E_{H}f= X_{H}f+Rf= df(X_{H})=\Omega(V,X_{H})=-(X_{H}\lrcorner\Omega)(V)\\
&=-dH(V)+(RH)\eta(V)=-VH+(RH)(V\lrcorner\eta)=0.
\end{split}
\end{equation}
So we have that $V=X_{f}$ and $f$ is a constant of motion with $Rf=0$.
\end{enumerate}
\end{proof}

As in the symplectic case, any infinitesimal symmetry generates a 
one-parameter group of canonical transformations leaving the Hamiltonian invariant.
To see this, we begin with the following definition from~\cite{AE2023}.
\begin{de}
\emph{A canonical transformation} for the cosymplectic manifold $(M,\Omega,\eta)$ is a 
possibly-local diffeomorphism $F$ on $M$ such that $F^{*}\Omega=\Omega$ and $F^{*}\eta=\eta$.
\end{de}
If $F:M\longrightarrow M$ is a diffeomorphism, then $(F^{*}\Omega,F^{*}\eta)$ is a cosymplectic structure on $M$ and therefore around any point $p\in M$ there are local coordinates $(Q^{1},\cdots,Q^{n},P_{1},\cdots,P_{n},T)$ such that
\begin{equation}
F^{*}\Omega=dQ^{i}\wedge dP_{i}\,,\qquad \textrm{and} \qquad F^{*}\eta=dT.
\end{equation}
If we think of $F$ locally as a coordinate transformation $(q^{1},\cdots,q^{n},p_{1},\cdots,p_{n},t)\longmapsto (Q^{1},\cdots,Q^{n},P_{1},\cdots,P_{n},t)$, then it is a canonical transformation if and only if the new coordinates $(Q^{1},\cdots,Q^{n},P_{1},\cdots,P_{n},t)$ are also canonical. 

Finally, we have the following property, which is completely analogous to the symplectic case (cf.~Proposition~\ref{prop:symm-canonical-sympl}).
\begin{prop}\label{prop:symm-canonical-cosympl}
If $V$ is an infinitesimal symmetry of $(M,\Omega,\eta, H)$, 
then the flow 
of $V$ defines a one-parameter group of canonical transformations
that leave the Hamiltonian invariant; 
reciprocally, if we have a one-parameter group of canonical transformations that leave the Hamiltonian $H$ invariant, then the infinitesimal generator of the flow defined by such one-parameter group is an infinitesimal symmetry of $(M,\Omega,\eta, H)$.
\end{prop}

\begin{ex}
\label{ex2}
Let us consider the Hamiltonian system $(M,\Omega,\eta,H)$ with $(q_{1},q_{2},p_{1},p_{2},t)$ 
canonical coordinates and $H(q^{1},q^{2},p_{1},p_{2},t)=\frac{1}{2}(p_{1}^{2}+p_{2}^{2})+tq_{1}$. 
The function $f(q_{1},q_{2},p_{1},p_{2},t)=p_{2}$ is a constant of motion with $R(f)=0$. 
Then $X_{f}$ is an infinitesimal symmetry of $(M,\Omega,\eta,H)$. Indeed
\begin{equation}
X_{f}=\frac{\partial}{\partial q_{2}},
\end{equation}
\begin{equation}
L_{X_{f}}\Omega=d(X_{f}\lrcorner \Omega)=d(dp_{2})=0, \hspace{1cm}X_{f}\lrcorner\eta=dt(\frac{\partial}{\partial q_{2}})=0
\end{equation}
and
\begin{equation}
L_{X_{f}}H=X_{f}H=\frac{\partial H}{\partial q_{2}}=0.
\end{equation}
\end{ex}

In the cosymplectic context we can also define infinitesimal dynamical symmetries as those vector fields that preserve the dynamics.
\begin{de}
A smooth vector field $W$ on $M$ is \emph{an infinitesimal dynamical symmetry} of $(M,\Omega,\eta,H)$ if $[W,E_{H}]=0$. 
\end{de}
Since $[W,E_{H}]=[W,X_{H}]+[W,R]$, we have that if $[W,X_{H}]=[W,R]=0$  then $W$ is an infinitesimal dynamical symmetry of $(M,\Omega,\eta,H)$. 
However, the converse statement is not true. 
By considering this kind of infinitesimal dynamical symmetries we have analogous results to the symplectic framework. 
As in the symplectic framework, we have that the assignment $f\mapsto X_{f}$ is a Lie algebra antihomomorphism between the Lie algebras $(C^{\infty}(M),\lbrace,\rbrace)$ and $(\mathfrak{X}(M),[,])$ \cite{CLL92}, i.e., $X_{\lbrace f,g\rbrace}=-[X_{f},X_{g}]$. In addition $X_{Rf}=-[X_{f},R]$. We have the following results.
\begin{prop}
Not every infinitesimal dynamical symmetry $W$ of $(M,\Omega,\eta,H)$ is such that $[W,X_{H}]=[W,R]=0$. 
\end{prop}
\begin{proof}
If $f$ is a constant of motion with $Rf\neq 0$ then 
\begin{equation}
[X_{f},E_{H}]=[X_{f},X_{H}]+[X_{f},R]=X_{\lbrace H,f\rbrace}-X_{Rf}=-X_{\lbrace f,H\rbrace+Rf}=0,
\end{equation}
i.e., $X_{f}$ is an infinitesimal dynamical symmetry, 
but $[X_{f},R]=-X_{Rf}$, which is not necessarily zero.
\end{proof}
\begin{prop}
Every infinitesimal symmetry $W$ of $(M,\Omega,\eta,H)$ is an infinitesimal dynamical symmetry.
\end{prop}
\begin{proof}
Let $W$ be an infinitesimal symmetry. Then $W=X_{f}$ for some constant of motion $f$ such that $Rf=0$. We have
\begin{equation}
[X_{f},X_{H}]=[X_{f},X_{H}]=X_{\lbrace H,f\rbrace}=0
\end{equation}
and
\begin{equation}
[X_{f},R]=-X_{Rf}=0\,.
\end{equation}
\end{proof}
There are non-Hamiltonian vector fields $W$ that satisfy $[W,X_{H}]=[W,R]=0$ and are not infinitesimal dynamical symmetries of $(M,\Omega,\eta,H)$. 
To show this, we introduce the concept of scaling symmetry in the cosymplectic framework.
\begin{de}
$X\in\mathfrak{X}(M)$ is \emph{a scaling symmetry of degree $\Lambda\in\mathbb{R}$}  for the Hamiltonian system $(M,\Omega,\eta,H)$ if
\begin{enumerate}
\item $L_{X}\Omega=\Omega$,
\item $X\lrcorner \eta=0$ (therefore $L_{X}\eta=0$) and
\item $L_{X}H=\Lambda H$
\end{enumerate}
\end{de}
As in the symplectic case, we have that scaling symmetries are not Hamiltonian for any function $f\in C^{\infty}(M)$, 
and thus they are not infinitesimal symmetries.
Moreover, if $X$ is a scaling symmetry of degree $\Lambda$, 
then $[X,X_{H}]=(\Lambda-1)X_{H}$. 
This is the content of the following result:
\begin{lem}\label{lem:scaling-cosympl}
If $X$ is a scaling symmetry of degree $\Lambda$ then $[X,X_{H}]=(\Lambda-1)X_{H}$.
\end{lem}
\begin{proof}
Let $X$ be a scaling symmetry. We have
\begin{equation}
[X,R]\lrcorner\eta=X\eta(R)-R\eta(X)=0
\end{equation}
and
\begin{equation}
[X,R]\lrcorner \Omega=L_{X}(R\lrcorner\Omega)-R\lrcorner L_{X}\Omega=0\,.
\end{equation}
Therefore $XR=RX$, and then
\begin{equation}
[X,X_{H}]\lrcorner\eta=X\eta(X_{H})-X_{H}\eta(X)=0
\end{equation}
and
\begin{equation}
\begin{split}
[X,X_{H}]\lrcorner \Omega &= L_{X}(X_{H}\lrcorner\Omega)-X_{H}\lrcorner L_{X}\Omega\\
&= L_{X}(dH-RH\eta)-X_{H}\lrcorner\Omega\\
&=L_{X}(dH)-XRH\eta-RHL_{X}\eta-dH+RH\eta\\
&=d(L_{X}H)-RXH\eta-dH+RH\eta\\
&=d((\Lambda-1)H)-R((\Lambda-1)H)\eta\,,
\end{split}
\end{equation}
which mean that
\begin{equation}
[X,X_{H}]=X_{(\Lambda-1)H}=(\Lambda-1)X_{H}.
\end{equation}
\end{proof}
Therefore we conclude that scaling symmetries of degree~1 are infinitesimal dynamical symmetries which satisfy $[X,X_{H}]=[X,R]=0$ but are not infinitesimal symmetries.

We also have a result analogous to Proposition 1 in~\cite{BJS2022} and which provides a useful characterization of scaling symmetries.
\begin{prop}\label{prop1}
$(M,\Omega,\eta,H)$ admits a scaling symmetry of degree $\Lambda$ if and only if there exists a primitive, $\lambda$, of $\Omega$ ($d\lambda=\Omega$), satisfying
\begin{equation}
X_{H}\lrcorner\lambda=-\Lambda H\hspace{1cm}\textit{and}\hspace{1cm}R\lrcorner\lambda=0.
\end{equation}
\end{prop}
\begin{proof}
Let $X$ be a scaling symmetry of degree $\Lambda$ for $(M,\Omega,\eta,H)$. 
Then for $\lambda=X\lrcorner\Omega$ we have
\begin{equation}
d\lambda=d(X\lrcorner\Omega)=L_{X}\Omega=\Omega,
\end{equation}
\begin{equation}
X_{H}\lrcorner\lambda=X_{H}\lrcorner(X\lrcorner\Omega)=\Omega(X,X_{H})=-dH(X)+(RH)\eta(X)=-XH=-\Lambda H
\end{equation}
and 
\begin{equation}
R\lrcorner\lambda=R\lrcorner(X\lrcorner\Omega)=\Omega(X,R)=0.
\end{equation}
Now, for the converse statement, let us suppose that there exists a primitive, $\lambda$, of $\Omega$ such that $X_{H}\lrcorner\lambda=-\Lambda H$ and $R\lrcorner\lambda=0$; since the map $\mathcal{X}_{\eta\Omega}:TM\longrightarrow T^{*}M$ defined by $\mathcal{X}_{\eta\Omega}(X)=X\lrcorner\Omega+(X\lrcorner\eta)\eta$ is a bundle isomorphism \cite{CLL92}, there exists a smooth vector field $X$ on $M$ such that 
\begin{equation}
X\lrcorner\Omega+(X\lrcorner\eta)\eta=\lambda.
\end{equation}
Since $R\lrcorner\lambda=0$ we have 
\begin{equation}
\begin{split}
R\lrcorner(X\lrcorner\Omega+(X\lrcorner\eta)\eta)=0 &\Longrightarrow \Omega(X,R)+(X\lrcorner\eta)\eta(R)=0\\
&\Longrightarrow X\lrcorner\eta=0.
\end{split}
\end{equation}
Now 
\begin{equation}
\begin{split}
d\lambda=\Omega &\Longrightarrow d(X\lrcorner\Omega)=\Omega\\
&\Longrightarrow L_{X}\Omega=\Omega,
\end{split}
\end{equation}
and
\begin{equation}
\begin{split}
X_{H}\lrcorner\lambda=-\Lambda H &\Longrightarrow \Omega(X,X_{H})=-\Lambda H\\
&\Longrightarrow -dH(X)+(RH)\eta(X)=-\Lambda H\\
&\Longrightarrow -XH=-\Lambda H\\
&\Longrightarrow XH=\Lambda H.
\end{split}
\end{equation}
So we conclude that $X$ is a scaling symmetry of degree $\Lambda$ for $(M,\Omega,\eta,H)$.
\end{proof} 

In the rest of this section we characterize the infinitesimal generators of one-parameter groups of canonoid transformations for cosymplectic Hamiltonian systems in a way completely analogous to the symplectic case. 
Canonoid transformations for time-dependent Hamiltonian systems have been studied in~\cite{CR89,TC2018,TCA2022}. 
The geometrical definition of canonoid transformations in the cosymplectic framework is presented in~\cite{AE2023} as follows:
\begin{de}
We say that a possibly-local diffeomorphism $F:M\longrightarrow M$ is \emph{a canonoid transformation} for the Hamiltonian system $(M,\Omega,\eta,H)$ if $F^{*}\eta=\eta$ (or equivalently $F_{\ast}R=R$) and there exists a function $K\in C^{\infty}(M)$ such that 
\begin{equation}
X_{H}\lrcorner F^{*}\Omega=dK- RK \eta.
\end{equation}
\end{de}
Locally, a canonoid transformation is a coordinate transformation of the form 
$$(q^{1},\cdots,q^{n},p_{1},\cdots,p_{n},t)\longmapsto (Q^{1},\cdots,Q^{n},P_{1},\cdots,P_{n},t)$$ 
such that the equations of motion in the coordinates $(Q^{1},\cdots,Q^{n},P_{1},\cdots,P_{n},t)$ are 
\begin{equation}
\dot{Q^{i}} = \frac{\partial K}{\partial P_{i}},\qquad
\dot{P_{i}} =-\frac{\partial K}{\partial Q^{i}},\qquad
\dot{t} =1\,,
\end{equation}
for some function $K\in C^{\infty}(M)$. 

The following result will be handy in order to obtain a characterization of the infinitesimal generators of canonoid transformations.
\begin{lem}
\label{teinffunc2}
$X$ is the infinitesimal generator of a one-parameter group of canonoid transformations if and only if $L_{X}\eta=0$ (or equivalently $L_{X}R=0$) and there exists (at least locally) a function $K$ on $M$ such that
\begin{equation}
\label{eqcanfunc2}
X_{H}\lrcorner L_{X}\Omega=dK-RK\eta. 
\end{equation}
\end{lem}
\begin{proof}
We have that $X$ is the infinitesimal generator of a one-parameter group of canonoid transformations if and only if $\varphi_{s}^{*}\eta=\eta$ and there exists a family $\lbrace K_{s}\rbrace$ of functions on $M$ such that
\begin{equation}
X_{H}\lrcorner \varphi_{s}^{*}\Omega=dK_{s}-RK_{s}\eta,
\end{equation}
which is equivalent to $d(X_{H}\lrcorner \varphi_{s}^{*}\Omega)=-d(RK_{s}\eta)$. 
Thus $X$ is the infinitesimal generator of a one-parameter group of canonoid transformations if and only if 
$\varphi_{s}^{*}\eta=\eta$ and
\begin{equation}
\label{eqinfcanonoid2}
L_{X_{H}}\varphi_{s}^{*}\Omega=-d(RK_{s}\eta).
\end{equation}
By taking $\frac{d}{ds}|_{s=0}$ in $\varphi_{s}^{*}\eta=\eta$ and in equation (\ref{eqinfcanonoid2}) we have that $X$ is the infinitesimal generator of a one-parameter group of canonoid transformations if and only if $L_{X}\eta=0$ and 
\begin{equation}
\label{eqRKcero}
L_{X_{H}}L_{X}\Omega=-d(Rk_{0}\eta),
\end{equation}
where $k_{0}=\frac{d}{ds}|_{s=0}K_{s}$. Equation (\ref{eqRKcero}) is equivalent to 
\begin{equation}
d(X_{H}\lrcorner L_{X}\Omega)=-d(Rk_{0}\eta),
\end{equation}
which is equivalent to the fact that there exists (at least locally) a function $K$ on $M$ such that
\begin{equation}
X_{H}\lrcorner L_{X}\Omega+Rk_{0}\eta=dK,
\end{equation}
or equivalently
\begin{equation}
X_{H}\lrcorner L_{X}\Omega=dK-Rk_{0}\eta.
\end{equation}
Finally, we have that $Rk_{0}=RK$. Indeed, 
\begin{equation}
RK-Rk_{0}=R\lrcorner(dK-Rk_{0}\eta)=R\lrcorner(X_{H}\lrcorner L_{X}\Omega)=L_{X}\omega(X_{H},R)=0.
\end{equation}
\end{proof}

In the cosymplectic case we know that the Hamiltonian function is not necessarily a constant of motion, namely, we have that $E_{H}H=RH$, although it is invariant under the flow of the Hamiltonian vector field $X_{H}$, i.e., $X_{H}H=0$. 
Therefore we have that the function $K$ in equation (\ref{eqcanfunc2}) is invariant under the flow of the Hamiltonian vector field $X_{H}$, and satisfies $E_{H}K=RK$. Indeed, 
\begin{equation}
X_{H}K=X_{H}\lrcorner(dK-RK\eta)=X_{H}\lrcorner(X_{H}\lrcorner L_{X}\Omega)=0.
\end{equation}

The following theorem provides a useful characterization of the infinitesimal generators of one-parameter groups of canonoid transformations, in complete analogy with Theorem~\ref{teinfcanonoid} for the symplectic case.
\begin{te}
\label{teinfcanonoid2}
$X$ is the infinitesimal generator of a one-parameter group of canonoid transformations for $(M,\Omega,\eta,H)$ if and only if $[X,R]=0$ and $[X,X_{H}]$ is a Hamiltonian vector field. In this case its Hamiltonian function is $L_{X}H-K$, where $K$ is given in equation~\eqref{eqcanfunc2}.
\end{te}
\begin{proof}
We have
\begin{equation}
[X,X_{H}]\lrcorner\eta=L_{X}(X_{H}\lrcorner\eta)-X_{H}\lrcorner L_{X}\eta=-X_{H}\lrcorner L_{X}\eta,
\end{equation}
and
\begin{equation}
\begin{split}
[X,X_{H}]\lrcorner\Omega &=L_{X}(X_{H}\lrcorner\Omega)-X_{H}\lrcorner L_{X}\Omega\\
&= L_{X}(dH-RH\eta)-X_{H}\lrcorner L_{X}\Omega\\
&= dL_{X}H-L_{X}(RH)\eta-(RH)L_{X}\eta-X_{H}\lrcorner L_{X}\Omega\\
&= dL_{X}H-RL_{X}H\eta-X_{H}\lrcorner L_{X}\Omega,
\end{split}
\end{equation}
so, using Lemma~\ref{teinffunc2}, we have that $X$ is the infinitesimal generator of a one-parameter group of canonoid transformations if and only if $L_{X}R=0$, $[X,X_{H}]\lrcorner\eta=0$ and there exists (at least locally) a function $K$ on $M$ such that
\begin{equation}
[X,X_{H}]\lrcorner\Omega= dL_{X}H-RL_{X}H\eta-dK+RK\eta=d(L_{X}H-K)-R(L_{X}H-K)\eta,
\end{equation}
which is equivalent to saying that $[X,R]=0$ and $[X,X_{H}]$ is a Hamiltonian vector field with Hamiltonian function $L_{X}H-K$.
\end{proof}

As in the symplectic case, we can now easily prove the following.
\begin{coro}\label{coroscalingcanonoid2}
Every scaling symmetry is the infinitesimal generator of a one-parameter group of canonoid transformations.
\end{coro}
\begin{proof}
We have that if $X$ is a scaling symmetry of $(M,\Omega,\eta, H)$, then $[X,X_{H}]=(\Lambda-1)X_{H}$ (see Lemma~\ref{lem:scaling-cosympl}). Therefore $[X,X_{H}]$ is a Hamiltonian vector field with Hamiltonian function $(\Lambda-1)H$. 
\end{proof}

Nevertheless, if $W$ is an infinitesimal dynamical symmetry of $(M,\Omega,\eta,H)$ and $\varphi$ is its flow, then we do not necessarily have that $\eta$ (nor $R$) is invariant under $\varphi$. So we do not necessarily have a group of canonoid transformations. 
Since we are interested in symmetries related to canonoid transformations, we consider a subset of the set of infinitesimal dynamical symmetries. 
We have
\begin{coro}
 If $W$ is such that $[W,X_{H}]=0$ and $[W,R]=0$, then $W$ is the infinitesimal generator of a one-parameter group of canonoid transformations.
\end{coro}
\begin{proof}
If $W$ is such that $[W,X_{H}]=0$ and $[W,R]=0$, then $[W,X_{H}]=0$ is a Hamiltonian vector field with constant Hamiltonian function $L_{X}H-K=c\in\mathbb{R}$ (we can see this in a completely analogous way as in the symplectic case). By hypothesis, $X_{H}$ and $R$ are invariant under the flow $\varphi$ of $W$. 
Therefore, we conclude that $W$ is the infinitesimal generator of a one-parameter group of canonoid transformations.
\end{proof}
    
In addition, we have that for each $s\in \mathbb R$, the new Hamiltonian function is $K_{s}=\varphi_{s}^{*}H$. 
Indeed, for every $Y\in\mathfrak{X}(M)$ we have
\begin{equation}
\begin{split}
(X_{H}\lrcorner \varphi_{s}^{*}\Omega)(Y) &=(\varphi_{s}^{*}\Omega)(X_{H},Y)\\
&=\Omega(\varphi_{s \ast}X_{H},\varphi_{s \ast}Y)\\
&=(X_{H}\lrcorner\Omega)(\varphi_{s \ast}Y)\\
&=\varphi_{s}^{*}(X_{H}\lrcorner\Omega)(Y)\\
&=(\varphi_{s}^{*}dH)(Y)-(\varphi_{s}^{*}(RH))(\varphi_{s}^{*}\eta)(Y)\\
&=(d\varphi_{s}^{*}H)(Y)-(\varphi_{s \ast}R)(\varphi_{s}^{*}H)\eta(Y)\\
&=(d\varphi_{s}^{*}H)(Y)-R(\varphi_{s}^{*}H)\eta(Y)
\end{split}
\end{equation}
i.e.~the vector field $X_{H}$ is Hamiltonian with respect to the cosymplectic structure $(\varphi_{s}^{*}\Omega,\eta)$ with Hamiltonian function $K_{s}=\varphi_{s}^{*}H$.

\section{Symmetries for contact Hamiltonian systems}
\label{sec4}

In this section we review some recent results concerning canonoid transformations in the contact setting and introduce new results regarding scaling symmetries in this geometry (we refer to Appendix~\ref{sec:appA} for more definitions and properties).

Let $(M,\theta,H)$ be a contact Hamiltonian system. Since these systems are
more adapted to describe dissipative systems in general, in~\cite{BG2021,GGMRR2020,LL2020} the concept of constant of motion (conserved quantity) has been generalized to the so-called dissipated quantities as follows:
\begin{de}
\emph{A dissipated quantity} of the contact Hamiltonian system $(M,\theta,H)$ is a function $f\in C^{\infty}(M)$ that is dissipated at the same rate as the contact Hamiltonian $H$, i.e., $X_{H}f=-fRH$.
\end{de}
In~\cite{BG2021} it is argued that a suitable definition of Noether symmetries in the contact setting is the following: 
\begin{de}
\emph{A Noether symmetry} of $(M,\theta,H)$ is a contact Hamiltonian vector field $X_{f}$ such that $\lbrace f,H\rbrace=0$. 
\end{de}
Then the contact version of Noether's theorem reads:
\begin{te}
\label{tecontact}
$X_{f}$ is a Noether symmetry if and only if $f$ is a dissipated quantity.
\end{te}
This theorem has been proved first in~\cite{GGMRR2020,LL2020}, where it was called \emph{the dissipation theorem}. 
Note that one of the reasons why dissipated quantities are important is that 
it is easy to show that the quotient of two dissipated quantities (whenever well defined) is a constant of motion.
Moreover, we observe that if $RH=0$ (locally in canonical coordinates this means that the Hamiltonian function does not depend on $z$), then dissipated quantities (including $H$ itself) are constants of motion. Indeed, in this case we have
\begin{equation}
X_{H}f=0\Leftrightarrow \lbrace H,f\rbrace=0.
\end{equation} 
This motivates the following definition from e.g.~\cite{Boyer2011,Visinescu2017}: 
\begin{de}
A contact Hamiltonian systems with $RH=0$ is called \emph{a good contact Hamiltonian system}.
\end{de}
Then it follows immediately that:
\begin{coro}
Let $(M,\theta,H)$ be a good contact Hamiltonian system. $X_{f}$ is a Noether symmetry if and only if $f$ is a constant of motion.
\end{coro}

Following the ideas developed in the previous sections, we are interested in symmetries of contact Hamiltonian systems related to canonical and canonoid transformations. We begin with the geometric definition of a canonical transformation in the contact setting as put forward in~\cite{AE2023}:
\begin{de}
\emph{A canonical transformation} for the contact manifold $(M,\theta)$ is a possibly-local diffeomorphism $F$ on $M$ such that $F^{*}\theta=\theta$.
\end{de}
Let $F:M\longrightarrow M$ be a diffeomorphism. We know that $F^{*}\theta$ is a contact structure on $M$. 
It follows that around any point $p\in M$ there are local coordinates $(Q^{1},\cdots,Q^{n},P_{1},\cdots,P_{n},Z)$ such that
\begin{equation}
F^{*}\theta=dZ-P_{i}dQ^{i}.
\end{equation}
If we think of $F$ locally as a coordinate transformation 
$$(q^{1},\cdots,q^{n},p_{1},\cdots,p_{n},z)\longmapsto (Q^{1},\cdots,Q^{n},P_{1},\cdots,P_{n},Z)\,,$$ 
then $F$ is a canonical transformation if and only if it preserves the Jacobi bracket (for details see~\cite{AE2023}). 

Analogously to the previous sections, we can now state the following definition from~\cite{GGMRR2020}:
\begin{de}\label{de-inf-symm-contact}
A smooth vector field $V$ on $M$ such that $L_{V}\theta=L_{V}H=0$ is called \emph{an infinitesimal (contact) symmetry} of the contact Hamiltonian system $(M,\theta,H)$.
\end{de}
Again, we have that:
\begin{prop}\label{prop:symm-canonical-contact}
If $V$ is an infinitesimal symmetry of $(M,\theta, H)$, then the flow 
of $V$ defines a one-parameter group of canonical transformations
that leave the Hamiltonian invariant; reciprocally, if we have a one-parameter group of canonical transformations that leave the Hamiltonian $H$ invariant, then the infinitesimal generator of the flow defined by such one-parameter group is an infinitesimal symmetry of $(M,\theta, H)$.
\end{prop}

Moreover, in this case infinitesimal symmetries are a subset of Noether symmetries, according to:
\begin{te}
$V$ is an infinitesimal symmetry of $(M,\theta,H)$ if and only if it is a Noether symmetry $X_{f}$ with $Rf=0$.
\end{te}
\begin{proof}
Let us suppose that $V$ is an infinitesimal symmetry, i.e., $L_{V}\theta=L_{V}H=0$. We have
\begin{equation}
L_{V}\theta=0\longrightarrow L_{V}d\theta=0\longrightarrow d(V\lrcorner d\theta)=0,
\end{equation}
i.e. there is a (possibly-local) function $f$ on $M$ such that $V\lrcorner d\theta=df$. 
Now
\begin{equation}
L_{V}\theta=0\longrightarrow d(V\lrcorner \theta)+V\lrcorner d\theta=0\longrightarrow V\lrcorner \theta=-f,
\end{equation}
so we have $V\lrcorner \theta=-f$ and $V\lrcorner d\theta=df$. On the other hand
\begin{equation}
L_{V}H=V(H)=dH(V)=d\theta(X_{H},V)+RH\theta(V)=-df(X_{H})-fRH=-X_{H}f-fRH,
\end{equation}
so
\begin{equation}
L_{V}H=0\longrightarrow X_{H}f=-fRH
\end{equation}
therefore $f$ is a dissipated quantity with $Rf=df(R)=d\theta(V,R)=-(R\lrcorner d\theta)(V)=0$, 
so we conclude that $V=X_{f}$ is a Noether symmetry with $Rf=0$.

Now let us see the converse statement. Let us suppose that $V$ is a 
Noether symmetry $X_{f}$ for some function $f$ such that $Rf=0$. Then
\begin{equation}
L_{X_{f}}\theta=d(X_{f}\lrcorner\theta)+X_{f}\lrcorner d\theta=-(Rf)\theta=0,
\end{equation}
\begin{equation}
L_{X_{f}}H=X_{f}H=\lbrace H,f\rbrace-HRf=-\lbrace f,H\rbrace=0.
\end{equation}
\end{proof}
Observe that if $(M,\theta,H)$ is a good contact Hamiltonian system, then we have an analogous result to the cosymplectic case, given by
\begin{coro}
Let $(M,\theta,H)$ be a good contact Hamiltonian system. 
Then  $V$ is an infinitesimal symmetry if and only if $V=X_{f}$ for some constant of motion $f$ such that $Rf=0$.
\end{coro}

We can proceed further to define infinitesimal dynamical symmetries in the contact framework~\cite{GGMRR2020}.
\begin{de}
A smooth vector field $W$ on $M$ is called \emph{an infinitesimal dynamical symmetry} of $(M,\theta,H)$ if $[W,X_{H}]=0$. 
\end{de}
In the contact case the assignment $f\mapsto X_{f}$ still defines a Lie algebra antihomomorphism 
between the Lie algebras $(C^{\infty}(M),\lbrace,\rbrace)$ and $(\mathfrak{X}(M),[,])$. 
Indeed, for $f,g\in C^{\infty}(M)$ we have
\begin{equation}
-[X_{f},X_{g}]=X_{\lbrace f,g\rbrace}
\end{equation}
For details see~\cite{LL2019}. 
Thus we have that every infinitesimal symmetry is an infinitesimal dynamical symmetry (this result was first shown in~\cite{GGMRR2020,GLR2022} in a different way from the one we present here).
\begin{prop}
Every infinitesimal symmetry is an infinitesimal dynamical symmetry.
\end{prop}
\begin{proof}
Let $V$ be an infinitesimal symmetry of $(M,\theta,H)$, i.e. $V=X_{f}$ for some function $f\in C^{\infty}(M)$ such that $\lbrace f,H\rbrace=Rf=0$. 
Then
\begin{equation}
-[X_{f},X_{H}]=X_{\lbrace f,H\rbrace}=0.
\end{equation}
Therefore $V$ is an infinitesimal dynamical symmetry of $(M,\theta,H)$.
\end{proof}
As in the previous sections, the converse statement is not true. 
To see that, we first introduce the concept of scaling symmetries in the contact framework and then we show that scaling symmetries of degree 1 are infinitesimal dynamical symmetries but not infinitesimal symmetries (they are not even Noether symmetries).
\begin{de}
$X\in\mathfrak{X}(M)$ is \emph{a scaling symmetry of degree $\Lambda\in\mathbb{R}$}  for the contact Hamiltonian system $(M,\theta,H)$ if
\begin{enumerate}
\item $L_{X}\theta=\theta$ (which implies that $L_{X}d\theta=d\theta$) and
\item $L_{X}H=\Lambda H$
\end{enumerate}
\end{de}
As in the previous sections, we obtain
\begin{lem}
If $X$ is a scaling symmetry of degree $\Lambda$ for $(M,\theta,H)$, then $[X,X_{H}]=(\Lambda-1)X_{H}$.
\end{lem}
\begin{proof}
Let $X$ be a scaling symmetry. We have 
\begin{equation}
[X,R]\lrcorner\theta=L_{X}(R\lrcorner\theta)-R\lrcorner L_{X}\theta=-R\lrcorner \theta=-1
\end{equation}
and
\begin{equation}
[X,R]\lrcorner d\theta=L_{X}(R\lrcorner d\theta)-R\lrcorner L_{X}d\theta=-R\lrcorner d\theta=0,
\end{equation}
and therefore $[X,R]=-R$. 
On the other hand
\begin{equation}
\begin{split}
[X,X_{H}]\lrcorner\theta &= L_{X}(X_{H}\lrcorner\theta)-X_{H}\lrcorner L_{X}\theta\\
&= -(L_{X}H+X_{H}\lrcorner L_{X}\theta)\\
&=-(\Lambda H-H)\\
&=-(\Lambda-1)H
\end{split}
\end{equation}
and
\begin{equation}
\begin{split}
[X,X_{H}]\lrcorner d\theta &= L_{X}(X_{H}\lrcorner d\theta)-X_{H}\lrcorner L_{X}d\theta\\
&= L_{X}(dH-RH\theta)-X_{H}\lrcorner dL_{X}\theta\\
&= dL_{X}H-XRH\theta-RHL_{X}\theta-X_{H}\lrcorner dL_{X}\theta\\
&= dL_{X}H+RH\theta-RXH\theta-RH\theta-dH+RH\theta\\
&= d((\Lambda-1)H)-R((\Lambda-1)H),
\end{split}
\end{equation}
and therefore
\begin{equation}
[X,X_{H}]=X_{(\Lambda-1)H}=(\Lambda-1)X_{H}.
\end{equation}
\end{proof}
So we have that scaling symmetries of degree 1 are examples of infinitesimal dynamical symmetries that are not infinitesimal symmetries. 
In this case however, a scaling symmetry might be a Hamiltonian vector field. To see this, observe that if $X_{f}$ is such that $R(f)=-1$ and $\lbrace H,f\rbrace=(\Lambda-1)H$, then $X_{f}$ is a scaling symmetry. 
Indeed,
\begin{equation}
L_{X_{f}}\theta=-(Rf)\theta=\theta
\end{equation}
and
\begin{equation}
L_{X_{f}}H=X_{f}H=\lbrace H,f\rbrace-HRf=(\Lambda-1)H+H=\Lambda H.
\end{equation}
To get more familiar with scaling symmetries for contact Hamiltonian systems, we provide now a couple of simple examples.
\begin{ex}
Let us consider the contact Hamiltonian system $(\mathbb R^3,\theta,H)$ with $(q,p,z)$ canonical coordinates and $H(q,p,z)=\frac{1}{2}p^{2}-\frac{1}{q}-\frac{1}{z^{2}}$. 
The vector field $X=2q\frac{\partial}{\partial q}-p\frac{\partial}{\partial p}+z\frac{\partial}{\partial z}$ is a scaling symmetry of degree $-2$. 
Indeed,
\begin{equation}
\begin{split}
L_{X}\theta &=d((2q\frac{\partial}{\partial q}-p\frac{\partial}{\partial p}+z\frac{\partial}{\partial z})\lrcorner (dz-pdq))+(2q\frac{\partial}{\partial q}-p\frac{\partial}{\partial p}+z\frac{\partial}{\partial z})\lrcorner dq\wedge dp\\
&=-2qdp-2pdq+dz+pdq+2qdp\\
&=dz-pdq\\
&=\theta
\end{split}
\end{equation}
and
\begin{equation}
L_{X}H=(2q\frac{\partial}{\partial q}-p\frac{\partial}{\partial p}+z\frac{\partial}{\partial z})(\frac{1}{2}p^{2}-\frac{1}{q}-\frac{1}{z^{2}})=-p^{2}+\frac{2}{q}+\frac{2}{z^{2}}=-2H.
\end{equation}
\end{ex}

\begin{ex}
Now let us consider the contact Hamiltonian system $(\mathbb R^3,\theta,H)$ with $(q,p,z)$ canonical coordinates and $H(q,p,z)=pf(q,z)+z$ where $f(q,z)$ is any smooth function on the variables $q,z$. The vector field $X=p\frac{\partial}{\partial p}+z\frac{\partial}{\partial z}$ is a scaling symmetry of degree $1$, and hence it is an infinitesimal dynamical symmetry. Indeed,
\begin{equation}
\begin{split}
L_{X}\theta &=d((p\frac{\partial}{\partial p}+z\frac{\partial}{\partial z})\lrcorner (dz-pdq))+(p\frac{\partial}{\partial p}+z\frac{\partial}{\partial z})\lrcorner dq\wedge dp\\
&=dz-pdq\\
&=\theta
\end{split}
\end{equation}
and
\begin{equation}
L_{X}H=(p\frac{\partial}{\partial p}+z\frac{\partial}{\partial z})(pf(q,z)+z)=pf(q,z)+z=H.
\end{equation}
In addition, we can see that $X$ is also a Hamiltonian vector field. Indeed, $X$ is the Hamiltonian vector field of the function $f=-z$.
\end{ex}

To finish this section, we characterize infinitesimal generators of one-parameter groups of canonoid transformations for contact Hamiltonian systems in a similar way as in the previous sections. The geometrical definition of canonoid transformations for contact Hamiltonian systems is presented in \cite{AE2023} as follows.
\begin{de}
We say that a possibly-local diffeomorphism $F:M\longrightarrow M$ is \emph{a canonoid transformation} for the Hamiltonian system $(M,\theta,H)$ if there exists a function $K\in C^{\infty}(M)$ such that 
\begin{equation}
X_{H}\lrcorner F^{*}\theta = -K \hspace{1cm}\textit{and}\hspace{1cm} X_{H}\lrcorner d F^{*}\theta =dK-(F_{\ast}R)KF^{*}\theta.
\end{equation}
\end{de}
Locally, a canonoid transformation is a coordinate transformation of the form 
$$(q^{1},\cdots,q^{n},p_{1},\cdots,p_{n},z)\longmapsto (Q^{1},\cdots,Q^{n},P_{1},\cdots,P_{n},Z)$$ 
such that the equations of motion in the coordinates $(Q^{1},\cdots,Q^{n},P_{1},\cdots,P_{n},Z)$ are 
\begin{equation}
\dot{Q^{i}} = \frac{\partial K}{\partial P_{i}},\qquad
\dot{P_{i}} =-\frac{\partial K}{\partial Q^{i}}-P_{i}\frac{\partial K}{\partial Z},\qquad
\dot{Z} = P_{i}\frac{\partial K}{\partial P_{i}}-K\,,
\end{equation}
for some function $K\in C^{\infty}(M)$.

We have the following result, analogous to Theorems~\ref{teinfcanonoid} and~\ref{teinfcanonoid2}.
\begin{te}
\label{teinfcanonoid3}
$X\in \mathfrak{X}(M)$ is the infinitesimal generator of a one-parameter group of canonoid transformations for $(M,\theta,H)$ if and only if $[X,X_{H}]$ is a Hamiltonian vector field. In this case its Hamiltonian function is $L_{X}H+X_{H}\lrcorner L_{X}\theta$.
\end{te}
\begin{proof}
Let $X\in \mathfrak{X}(M)$ and $\varphi$ its flow. $X$ is the infinitesimal generator of a one-parameter group of canonoid transformations for $(M,\theta,H)$ if and only if 
\begin{equation}
\label{eqcanonoids}
X_{H}\lrcorner \varphi_{s}^{*}\theta = -K_{s} \hspace{1cm}\textit{and}\hspace{1cm} X_{H}\lrcorner d \varphi_{s}^{*}\theta =dK_{s}-(\varphi_{s\ast}R)K_{s}\varphi_{s}^{*}\theta
\end{equation}
for a family $\lbrace K_{s}\rbrace$ of functions on $M$. By taking $\frac{d}{ds}|_{s=0}$ in equation $X_{H}\lrcorner d \varphi_{s}^{*}\theta =dK_{s}-(\varphi_{s\ast}R)K_{s}\varphi_{s}^{*}\theta$ we have that it is equivalent to
\begin{equation}
X_{H}\lrcorner dL_{X}\theta=dK-(L_{X}R)H\theta-RK\theta-RHL_{X}\theta
\end{equation}
where $K=\frac{d}{ds}|_{s=0}K_{s}$ and $K_{s}|_{s=0}=H$. Now by taking $\frac{d}{ds}|_{s=0}$ in equation $X_{H}\lrcorner \varphi_{s}^{*}\theta = -K_{s}$ we have that 
\begin{equation}
K=-X_{H}\lrcorner L_{x}\theta.
\end{equation}
So $X$ is the infinitesimal generator of a one-parameter group of canonoid transformations for $(M,\theta,H)$ if and only if
\begin{equation}
X_{H}\lrcorner dL_{X}\theta=-d(X_{H}\lrcorner L_{x}\theta)-(L_{X}R)H\theta+R(X_{H}\lrcorner L_{x}\theta)\theta-RHL_{X}\theta,
\end{equation}
or equivalently 
\begin{equation}
\label{eqcanfunc3}
-d(X_{H}\lrcorner L_{x}\theta)-(L_{X}R)H\theta+R(X_{H}\lrcorner L_{x}\theta)\theta-RHL_{X}\theta-X_{H}\lrcorner dL_{X}\theta=0.
\end{equation}

On the other hand, we have that
\begin{equation}
\begin{split}
[X,X_{H}]\lrcorner\theta &= L_{X}(X_{H}\lrcorner\theta)-X_{H}\lrcorner L_{X}\theta\\
&= -(L_{X}H+X_{H}\lrcorner L_{X}\theta)
\end{split}
\end{equation}
and
\begin{equation}
\label{eqhamveccan3}
\begin{split}
[X,X_{H}]\lrcorner d\theta &= L_{X}(X_{H}\lrcorner d\theta)-X_{H}\lrcorner L_{X}d\theta\\
&= L_{X}(dH-RH\theta)-X_{H}\lrcorner dL_{X}\theta\\
&= dL_{X}H-XRH\theta-RHL_{X}\theta-X_{H}\lrcorner dL_{X}\theta\\
&=d(L_{X}H+X_{H}\lrcorner L_{X}\theta)-R(L_{X}H+X_{H}\lrcorner L_{X}\theta)\theta\\
&-d(X_{H}\lrcorner L_{x}\theta)-(L_{X}R)H\theta+R(X_{H}\lrcorner L_{x}\theta)\theta-RHL_{X}\theta-X_{H}\lrcorner dL_{X}\theta.
\end{split}
\end{equation}
So, by using equations (\ref{eqcanfunc3}) and (\ref{eqhamveccan3}), 
we conclude that $X$ is the infinitesimal generator of a one-parameter group of canonoid transformations for $(M,\theta,H)$ if and only if 
\begin{equation}
[X,X_{H}]\lrcorner\theta = -(L_{X}H+X_{H}\lrcorner L_{X}\theta)
\end{equation}
and
\begin{equation}
[X,X_{H}]\lrcorner d\theta=d(L_{X}H+X_{H}\lrcorner L_{X}\theta)-R(L_{X}H+X_{H}\lrcorner L_{X}\theta)\theta,
\end{equation}
i.e., $X$ is the infinitesimal generator of a one-parameter group of canonoid transformations for $(M,\theta,H)$ if and only if $[X,X_{H}]$ is a Hamiltonian vector field with Hamiltonian function $L_{X}H+X_{H}\lrcorner L_{X}\theta$.
\end{proof}

As in the previous sections, we are interested in the following corollary:
\begin{coro}\label{coroscalingcanonoid3}
Scaling symmetries are infinitesimal generators of one-parameter groups of canonoid transformations.
\end{coro}
\begin{proof}
If $X$ is a scaling symmetry of $(M,\theta,H)$ then $[X,X_{H}]=(\Lambda-1)X_{H}$ (where $\Lambda$ is the degree of $X$), which is the Hamiltonian vector field of the function $(\Lambda-1)H$.
\end{proof}
 Of course, we also have 
 \begin{coro}
If $X$ is an infinitesimal dynamical symmetry of $(M,\theta,H)$, then the flow of $X$ defines a one-parameter group of canonoid transformations.
 \end{coro}

In addition, we have that for each $s\in\mathbb R$, the new Hamiltonian function is $K_{s}=\varphi_{s}^{*}H$. 
Indeed, let us suppose that $X$ is an infinitesimal dynamical symmetry of $(M,\theta,H)$, then $\varphi_{s \ast}X_{H}=X_{H}$, so that
\begin{equation}
(X_{H}\lrcorner \varphi_{s}^{*}\theta)(p)=\theta_{\varphi_{s}(p)}(\varphi_{s \ast}X_{H})=\theta_{\varphi_{s}(p)}(X_{H})=-H(\varphi_{s}(p))=-(\varphi_{s}^{*}H)(p)
\end{equation}
and for every $Y\in\mathfrak{X}(M)$ we have
\begin{equation}
\begin{split}
(X_{H}\lrcorner \varphi_{s}^{*}d\theta)(Y) &=(\varphi_{s}^{*}d\theta)(X_{H},Y)\\
&=d\theta(\varphi_{s \ast}X_{H},\varphi_{s \ast}Y)\\
&=(X_{H}\lrcorner d\theta)(\varphi_{s \ast}Y)\\
&=\varphi_{s}^{*}(X_{H}\lrcorner d\theta)(Y)\\
&=(\varphi_{s}^{*}dH-(\varphi_{s \ast}R)(\varphi_{s}^{*}H)\varphi_{s}^{*}\theta)(Y)\\
&=(d\varphi_{s}^{*}H-(\varphi_{s \ast}R)(\varphi_{s}^{*}H)\varphi_{s}^{*}\theta)(Y)
\end{split}
\end{equation}
i.e.~the vector field $X_{H}$ is Hamiltonian with respect to the contact structure $\varphi_{s}^{*}\theta$ with Hamiltonian function $K_{s}=\varphi_{s}^{*}H$.

In~\cite{AE2023} it has been shown that from a given canonoid transformation for a \emph{good} contact Hamiltonian system one can obtain
the corresponding constants of motion.
It is worth remarking that in our case the contact Hamiltonian system needs not be good.
Therefore in the general case we have that if $\varphi_{s}$ is a one-parameter group of canonoid transformations for $(M,\theta,H)$ with new Hamiltonian functions $K_{s}$, then $K_{s}$ are dissipated quantities; therefore the quotients $K_{s}/K_{s^{\prime}}$ (whenever well defined) are constants of motion.

\section{Symmetries for cocontact Hamiltonian systems}
\label{sec5}

In~\cite{BG2021} the contact version of Noether's theorem has been extended to time-dependent contact Hamiltonian systems using the extended phase space $T^{*}Q\times\mathbb{R}\times\mathbb{R}$; on the other hand, time-dependent contact Hamiltonian systems can be generalized using the formalism of cocontact geometry, to obtain the so-called cocontact Hamiltonian systems~\cite{Letal2022}. In~\cite{GLR2022} the authors studied dissipated quantities and dynamical symmetries for cocontact Hamiltonian systems.
In this section, we review some recent aspects concerning canonoid transformations in the cocontact setting and introduce new results (we refer to Appendix~\ref{sec:appA} for more definitions and properties).

Let $(M,\theta,\eta,H)$ be a cocontact Hamiltonian system. 
\begin{de}
\emph{A Noether symmetry} of $(M,\theta,\eta,H)$ is a contact Hamiltonian vector field $X_{f}$ such that $\lbrace f,H\rbrace+R_{t}f=0$. 
\end{de}
\begin{de}
\emph{A dissipated quantity} of the cocontact Hamiltonian system $(M,\theta,\eta,H)$ is a function $f\in C^{\infty}(M)$ such that $E_{H}f=-fR_{z}H$.
\end{de}
As it was pointed out in \cite{BG2021,GLR2022}, in the most general case we have $E_{H}H=-HR_{z}H+R_{t}H$,
and therefore whenever $H$ depends explicitly on $t$ it is not itself a dissipated quantity. 

As expected, we have a cocontact version of Noether's theorem (see Proposition 3.6 in~\cite{GLR2022}):
\begin{te}
A function $f\in C^{\infty}(M)$ is a dissipated quantity if and only if $\lbrace f,H\rbrace+R_{t}f=0$, which is equivalent to $X_{f}$ being a Noether symmetry.
\end{te} 
As in the contact case, we have that the quotient of two dissipated quantities is a constant of motion~\cite{GLR2022}.
Now, we can also observe that if $R_{z}H=0$ (locally, the Hamiltonian function does not depend on $z$), then dissipated quantities are constants of motion. 
Following the notion of good contact Hamiltonian systems, we define the following:
\begin{de}
We say that $(M,\theta,\eta,H)$ is \emph{a good cocontact Hamiltonian system} if $R_{z}H=0$.
\end{de}
It is important to remark that the condition $R_{t}H=0$ is not required, i.e., the Hamiltonian function $H$ may depend explicitly on $t$. 
Then we have the following immediate corollary: 
\begin{coro}
Let $(M,\theta,\eta,H)$ be a good cocontact Hamiltonian system. Then $X_{f}$ is a Noether symmetry if and only if $f$ is a constant of motion.
\end{coro}

Again, we are interested in symmetries of cocontact Hamiltonian systems which are related to canonical and canonoid transformations.
We start with the following definition from~\cite{AE2023}.
\begin{de}
\emph{A canonical transformation} for the cocontact manifold $(M,\theta,\eta)$ is a possibly-local 
diffeomorphism $F$ on $M$ such that $F^{*}\theta=\theta$ and $F^{*}\eta=\eta$.
\end{de}
Let $F:M\longrightarrow M$ be a diffeomorphism. We know that $(F^{*}\theta,F^{*}\eta)$ is a cocontact structure on $M$.
Then around any point $p\in M$ there are local coordinates $(T,Q^{1},\cdots,Q^{n},P_{1},\cdots,P_{n},Z)$ such that
\begin{equation}
F^{*}\theta=dZ-P_{i}dQ^{i}\hspace{1cm}\textit{and}\hspace{1cm}F^{*}\eta=dT.
\end{equation}
If we think of $F$ locally as a coordinate transformation 
$$(t,q^{1},\cdots,q^{n},p_{1},\cdots,p_{n},z)\longmapsto (t,Q^{1},\cdots,Q^{n},P_{1},\cdots,P_{n},Z)\,,$$ 
then it is a canonical transformation if and only if it preserves the Jacobi bracket (see~\cite{AE2023} for the details). 

We are interested in vector fields whose flow leave the cocontact structure and the Hamiltonian invariant. 
In~\cite{GLR2022} these vector fields are called infinitesimal strict Hamiltonian symmetries; by following the language employed in the previous sections, we call them infinitesimal symmetries.
\begin{de}
A smooth vector field $V$ on $M$ is \emph{an infinitesimal symmetry} of $(M,\theta,\eta,H)$ if $L_{V}\theta=L_{V}H=0$ 
and $V\lrcorner\eta=0$ (which implies $L_{V}\eta=0$).
\end{de}
We have some results similar to the previous sections.
\begin{prop}\label{prop:symm-canonical-cocontact}
If $V$ is an infinitesimal symmetry of $(M,\theta,\eta, H)$, 
then the flow 
of $V$ defines a one-parameter group of canonical transformations
that leave the Hamiltonian invariant; reciprocally, if we have a one-parameter group of canonical transformations that leave the Hamiltonian $H$ 
invariant, then the infinitesimal generator of the flow defined by such one-parameter group is an infinitesimal symmetry of $(M,\theta,\eta, H)$.
\end{prop}

Moreover, in this case infinitesimal symmetries are a subset of Noether symmetries, according to:
\begin{te}
$V$ is an infinitesimal symmetry if and only if it is a Noether symmetry $X_{f}$ with $R_{z}f=R_{t}f=0$.
\end{te}
\begin{proof}
Let us suppose that $V$ is an infinitesimal symmetry, i.e. $L_{V}\theta=L_{V}H=0$ and $V\lrcorner\eta=0$. We have
\begin{equation}
L_{V}\theta=0\longrightarrow L_{V}d\theta=0\longrightarrow d(V\lrcorner d\theta)=0,
\end{equation}
i.e. there is a (possibly-local) function $f$ on $M$ such that $V\lrcorner d\theta=df$. 
Now
\begin{equation}
L_{V}\theta=0\longrightarrow d(V\lrcorner \theta)+V\lrcorner d\theta=0\longrightarrow V\lrcorner \theta=-f,
\end{equation}
so we have $V\lrcorner \theta=-f$ and $V\lrcorner d\theta=df$. Now we can see that for $R=R_{z}$ or $R=R_{t}$ we have
\begin{equation}
Rf=df(R)=d\theta(V,R)=-(R\lrcorner d\theta)(V)=0.
\end{equation}
On the other hand
\begin{equation}
L_{V}H=V(H)=dH(V)=d\theta(X_{H},V)+R_{z}H\theta(V)=-df(X_{H})-fR_{z}H=-X_{H}f-fR_{z}H,
\end{equation}
so
\begin{equation}
L_{V}H=0\longrightarrow X_{H}f=-fR_{z}H\longrightarrow E_{H}f=-fR_{z}H
\end{equation}
therefore $f$ is a dissipated quantity, so we conclude that $V=X_{f}$ is a Noether symmetry with $R_{z}f=R_{t}f=0$.

Now let us see the converse statement. Let us suppose that $V$ is a Noether symmetry $X_{f}$ 
for some function $f$ such that $R_{z}f=R_{t}f=0$. Then
\begin{equation}
L_{X_{f}}\theta=d(X_{f}\lrcorner\theta)+X_{f}\lrcorner d\theta=-(R_{z}f)\theta-R_{t}f\eta=0,
\end{equation}
\begin{equation}
L_{X_{f}}H=X_{f}H=\lbrace H,f\rbrace-HR_{z}f=-\lbrace f,H\rbrace=0
\end{equation}
and of course $X_{f}\lrcorner\eta=0$.
\end{proof}

Now we consider the definition of infinitesimal dynamical symmetries in the cocontact context (cf.~\cite{GLR2022}).
\begin{de}
A smooth vector field $W$ on $M$ is \emph{an infinitesimal dynamical symmetry} of $(M,\theta,\eta,H)$ if $[W,E_{H}]=0$.
\end{de}
 
As in the previous cases, the assignment $f\mapsto X_{f}$ defines a Lie algebra antihomomorphism between the Lie algebras $(C^{\infty}(M),\lbrace,\rbrace)$ and $(\mathfrak{X}(M),[,])$.
Let us recall also that in the cocontact framework for $f,g\in C^{\infty}(M)$ we have 
\begin{equation}
-[X_{f},X_{g}]\lrcorner\theta=-\lbrace f,g\rbrace
\end{equation}
and in addition
\begin{equation}
-[X_{f},X_{g}]\lrcorner\eta=0,
\end{equation}
so
\begin{equation}
-[X_{f},X_{g}]=X_{\lbrace f,g\rbrace}.
\end{equation}
Analogously to the cosymplectic case, we can see that
\begin{equation}
\label{eqXRz}
-[X_{f},R_{z}]=X_{R_{z}f}
\end{equation}
and 
\begin{equation}
\label{eqXRt}
-[X_{f},R_{t}]=X_{R_{t}f}.
\end{equation}
We have the following results.
\begin{prop}
Every infinitesimal symmetry $V$ of $(M,\theta,\eta,H)$ satisfies $[V,X_{H}]=[V,R_{t}]=0$. 
Therefore $V$ is an infinitesimal dynamical symmetry.
\end{prop}
\begin{proof}
Let $V$ be an infinitesimal symmetry of $(M,\theta,\eta,H)$, i.e.~$V=X_{f}$ for some function $f\in C^{\infty}(M)$ 
such that $\lbrace f,H\rbrace=R_{z}f=R_{t}f=0$. Then 
\begin{equation}
-[X_{f},R_{t}]=X_{R_{t}f}=0
\end{equation}
and
\begin{equation}
-[X_{f},X_{H}]=X_{\lbrace f,H\rbrace}=0.
\end{equation}
\end{proof}

However, the converse statement is not true. To show this, we start with the following:
\begin{prop}
Not every infinitesimal dynamical symmetry $W$ of $(M,\theta,\eta,H)$ is such that $[W,X_{H}]=[W,R_{t}]=0$. 
\end{prop}
\begin{proof}
If $f$ is a dissipated quantity with $R_{t}f\neq 0$, then 
\begin{equation}
[X_{f},E_{H}]=[X_{f},X_{H}]+[X_{f},R_{t}]=X_{\lbrace H,f\rbrace}-X_{R_{t}f}=X_{\lbrace f,H\rbrace+R_{t}f}=0,
\end{equation}
i.e.~$X_{f}$ is an infinitesimal dynamical symmetry. However, in general $[X_{f},R_{t}]=X_{R_{t}f}$, 
which is not necessarily zero.
\end{proof}

Now we introduce the concept of scaling symmetries in the cocontact framework. 
\begin{de}
$X\in\mathfrak{X}(M)$ is \emph{a scaling symmetry of degree $\Lambda\in\mathbb{R}$}  
for the cocontact Hamiltonian system $(M,\theta,\eta,H)$ if
\begin{enumerate}
\item $L_{X}\theta=\theta$ (therefore $L_{X}d\theta=d\theta$),
\item $X\lrcorner\eta=0$ (therefore $L_{X}\eta=0$) and
\item $L_{X}H=\Lambda H$
\end{enumerate}
\end{de}
As in the previous cases, we have:
\begin{lem}
If $X$ is a scaling symmetry of degree $\Lambda$ for $(M,\theta,\eta,H)$ then $[X,X_{H}]=(\Lambda-1)X_{H}$.
\end{lem}
\begin{proof}
Let $X$ be a scaling symmetry. By a calculation analogous to the ones in the previous sections, we have: 
\begin{equation}
[X,R_{z}]=-R_{z},
\end{equation}
\begin{equation}
[X,R_{t}]=0,
\end{equation}
\begin{equation}
[X,X_{H}]\eta=0,
\end{equation}
\begin{equation}
[X,X_{H}]\lrcorner\theta=-(\Lambda-1)H
\end{equation}
and
\begin{equation}
[X,X_{H}]\lrcorner d\theta=d(\Lambda H-H)-R(\Lambda H-H)\theta.
\end{equation}
Thus, we conclude that
\begin{equation}
[X,X_{H}]=X_{\Lambda H-H}=\Lambda X_{H}-X_{H}=(\Lambda-1)X_{H}.
\end{equation}
\end{proof}
So we have that scaling symmetries of degree 1 are examples of infinitesimal dynamical symmetries
such that $[X,X_{H}]=[X,R_{t}]=0$, but they are not infinitesimal symmetries. 
As in the contact case, a scaling symmetry might be a Hamiltonian vector field. 
To see this, we observe that if $X_{f}$ is such that $R_{z}(f)=-1$, $R_{t}f=0$ and $\lbrace H,f\rbrace=(\Lambda-1)H$, then $X_{f}$ is a scaling symmetry. 
Indeed,
\begin{equation}
L_{X_{f}}\theta=-(R_{z}f)\theta-(R_{t}f)\eta=\theta
\end{equation}
and
\begin{equation}
L_{X_{f}}H=X_{f}H=\lbrace H,f\rbrace-HR_{z}f=(\Lambda-1)H+H=\Lambda H.
\end{equation}

The geometrical definition of canonoid transformations for cocontact Hamiltonian systems has been introduced in~\cite{AE2023} as follows.
\begin{de}
We say that a possibly-local diffeomorphism $F:M\longrightarrow M$ is \emph{a canonoid transformation} for the Hamiltonian system $(M,\theta,\eta,H)$ if $F^{*}\eta=\eta$ (or equivalently $F_{\ast}R_{t}=R_{t}$) and there exists a function $K\in C^{\infty}(M)$ such that 
\begin{equation}
X_{H}\lrcorner F^{*}\theta = -K \hspace{1cm}\textit{and}\hspace{1cm} X_{H}\lrcorner d F^{*}\theta =dK-(F_{\ast}R_{z})KF^{*}\theta-R_{t}K\eta.
\end{equation}
\end{de}
Locally, a canonoid transformation is a coordinate transformation of the form 
$$(q^{1},\cdots,q^{n},p_{1},\cdots,p_{n},z,t)\longmapsto (t,Q^{1},\cdots,Q^{n},P_{1},\cdots,P_{n},Z,t)$$ 
such that the equations of motion in the coordinates $(Q^{1},\cdots,Q^{n},P_{1},\cdots,P_{n},Z,t)$ are 
\begin{equation}
\dot{Q^{i}} = \frac{\partial K}{\partial P_{i}},\qquad
\dot{P_{i}} =- \frac{\partial K}{\partial Q^{i}}-P_{i}\frac{\partial K}{\partial Z}\,,\qquad
\dot{Z} = P_{i}\frac{\partial K}{\partial P_{i}}-K\,,\qquad
\dot{t}=1\,,
\end{equation}
for some function $K\in C^{\infty}(M)$. 

We have the following characterization result, analogous to Theorems~\ref{teinfcanonoid},~\ref{teinfcanonoid2} and~\ref{teinfcanonoid3}.
\begin{te}
\label{teinfcanonoid4}
$X\in \mathfrak{X}(M)$ is the infinitesimal generator of a one-parameter group of canonoid transformations for $(M,\theta.\eta,H)$ if and only if $[X,R_{t}]=0$ and $[X,X_{H}]$ is a Hamiltonian vector field. 
In this case its Hamiltonian function is $L_{X}H+X_{H}\lrcorner L_{X}\theta$.
\end{te}
\begin{proof}
Let $X\in \mathfrak{X}(M)$ and $\varphi$ its flow. $X$ is the infinitesimal generator of a one-parameter group of canonoid transformations for $(M,\theta,\eta,H)$ if and only if 
$\varphi_{s}^{*}\eta=\eta$ (or equivalently $\varphi_{s \ast}R_{t}=R_{t}$),
\begin{equation}
\label{eqcocanonoids}
X_{H}\lrcorner \varphi_{s}^{*}\theta = -K_{s} \hspace{1cm}\textit{and}\hspace{1cm} X_{H}\lrcorner d \varphi_{s}^{*}\theta =dK_{s}-(\varphi_{s\ast}R_{z})K_{s}\varphi_{s}^{*}\theta-R_{t}K_{s}\eta
\end{equation}
for a family $\lbrace K_{s}\rbrace$ of functions on $M$. By taking $\frac{d}{ds}|_{s=0}$ in $\varphi_{s}^{*}\eta=\eta$ we have $L_{X}\eta=0$ (or equivalently $L_{X}R_{t}=0$). Now by taking $\frac{d}{ds}|_{s=0}$ in equation $X_{H}\lrcorner d \varphi_{s}^{*}\theta = dK_{s}-(\varphi_{s\ast}R_{z})K_{s}\varphi_{s}^{*}\theta-R_{t}K_{s}\eta$ we have that it is equivalent to
\begin{equation}
X_{H}\lrcorner dL_{X}\theta=dK-(L_{X}R_{z})H\theta-R_{z}K\theta-R_{z}HL_{X}\theta--R_{t}K\eta
\end{equation}
where $K=\frac{d}{ds}|_{s=0}K_{s}$ and $K_{s}|_{s=0}=H$. 
In addition, by taking $\frac{d}{ds}|_{s=0}$ in equation $X_{H}\lrcorner \varphi_{s}^{*}\theta = -K_{s}$ we have that 
\begin{equation}
K=-X_{H}\lrcorner L_{x}\theta.
\end{equation}
So $X$ is the infinitesimal generator of a one-parameter group of canonoid transformations for $(M,\theta,H)$ if and only if $L_{X}R_{t}=0$ and
\begin{equation}
X_{H}\lrcorner dL_{X}\theta=-d(X_{H}\lrcorner L_{x}\theta)-(L_{X}R_{z})H\theta+R_{z}(X_{H}\lrcorner L_{x}\theta)\theta-R_{z}HL_{X}\theta+R_{t}(X_{H}\lrcorner L_{X}\theta)\eta,
\end{equation}
or equivalently $[X,R_{t}]=0$ and
\begin{equation}
\label{eqcanfunc4}
-d(X_{H}\lrcorner L_{x}\theta)-(L_{X}R_{z})H\theta+R_{z}(X_{H}\lrcorner L_{x}\theta)\theta-R_{z}HL_{X}\theta+R_{t}(X_{H}\lrcorner L_{X}\theta)\eta-X_{H}\lrcorner dL_{X}\theta=0.
\end{equation}

On the other hand, we have that
\begin{equation}
[X,X_{H}]\lrcorner\eta=L_{X}(X_{H}\lrcorner\eta)-X_{H}\lrcorner L_{X}\eta,
\end{equation}
\begin{equation}
\begin{split}
[X,X_{H}]\lrcorner\theta &= L_{X}(X_{H}\lrcorner\theta)-X_{H}\lrcorner L_{X}\theta\\
&= -(L_{X}H+X_{H}\lrcorner L_{X}\theta)
\end{split}
\end{equation}
and
\begin{equation}
\label{eqhamveccan4}
\begin{split}
[X,X_{H}]\lrcorner d\theta &= L_{X}(X_{H}\lrcorner d\theta)-X_{H}\lrcorner L_{X}d\theta\\
&= L_{X}(dH-R_{z}H\theta-R_{t}H\eta)-X_{H}\lrcorner dL_{X}\theta\\
&= dL_{X}H-XR_{z}H\theta-R_{z}HL_{X}\theta-XR_{t}H\eta-R_{t}HL_{X}\eta-X_{H}\lrcorner dL_{X}\theta\\
&=d(L_{X}H+X_{H}\lrcorner L_{X}\theta)-R_{z}(L_{X}H+X_{H}\lrcorner L_{X}\theta)\theta-R_{t}(L_{X}H+X_{H}\lrcorner L_{X}\theta)\eta\\
&-d(X_{H}\lrcorner L_{x}\theta)-(L_{X}R_{z})H\theta+R_{z}(X_{H}\lrcorner L_{x}\theta)\theta-R_{z}HL_{X}\theta+R_{t}(X_{H}\lrcorner L_{X}\theta)\eta\\
&-(L_{X}R_{t})H\eta-R_{t}HL_{X}\eta-X_{H}\lrcorner dL_{X}\theta.
\end{split}
\end{equation}
So, by using equations (\ref{eqcanfunc4}) and (\ref{eqhamveccan4}), 
we conclude that $X$ is the infinitesimal generator of a one-parameter group of canonoid transformations for $(M,\theta,H)$ if and only if $[X,R_{t}]=0$,
\begin{equation}
[X,X_{H}]\lrcorner\eta=0,
\end{equation}
\begin{equation}
[X,X_{H}]\lrcorner\theta = -(L_{X}H+X_{H}\lrcorner L_{X}\theta)
\end{equation}
and
\begin{equation}
[X,X_{H}]\lrcorner d\theta=d(L_{X}H+X_{H}\lrcorner L_{X}\theta)-R_{z}(L_{X}H+X_{H}\lrcorner L_{X}\theta)\theta-R_{t}(L_{X}H+X_{H}\lrcorner L_{X}\theta)\eta,
\end{equation}
i.e., $X$ is the infinitesimal generator of a one-parameter group of canonoid transformations for $(M,\theta,H)$ if and only if $[X,R_{t}]=0$  and $[X,X_{H}]$ is a Hamiltonian vector field with Hamiltonian function $L_{X}H+X_{H}\lrcorner L_{X}\theta$.
\end{proof}

Finally, we also have
\begin{coro}\label{coroscalingcanonoid4}
Scaling symmetries are infinitesimal generators of one-parameter groups of canonoid transformations.
\end{coro}
\begin{proof}
If $X$ is a scaling symmetry of $(M,\theta,\eta,H)$ then $[X,X_{H}]=(\Lambda-1)X_{H}$ (where $\Lambda$ is the degree of $X$), which is a Hamiltonian vector field with Hamiltonian function $(\Lambda-1)H$.
\end{proof}
However, as in the cosymplectic case, we have that if $W$ is an infinitesimal dynamical symmetry of $(M,\theta,\eta,H)$, then its flow does not necessarily preserve $\eta$ (nor $R_{t}$). 
Therefore, it does not necessarily define a group of canonoid transformations. 
Since we are interested in symmetries related to canonoid transformations, we thus consider a subset of the set of infinitesimal dynamical symmetries.
\begin{coro}
If $W$ is such that $[W,X_{H}]=0$ and $[W,R_{t}]=0$, then its flow defines a one-parameter group of canonoid transformations.
\end{coro}
\begin{proof}
If $W$ is such that $[W,X_{H}]=0$ and $[W,R_{t}]=0$, 
then $[W,X_{H}]=0$ is a Hamiltonian vector field with Hamiltonian function $L_{X}H+X_{H}\lrcorner L_{X}\theta=-[X,X_{H}]\lrcorner\theta=0$. 
So the flow $\varphi$ of $W$ defines a one-parameter group of canonoid transformations.
\end{proof}

In addition, we have that for each $s\in\mathbb R$, the new Hamiltonian function is 
$K_{s}=\varphi_{s}^{*}H$. Indeed, $\varphi_{s}^{*}X_{H}=X_{H}$ and $\varphi_{s \ast}R_{t}=R_{t}$ and we have
\begin{equation}
(X_{H}\lrcorner \varphi_{s}^{*}\theta)(p)=\theta_{\varphi_{s}(p)}(\varphi_{s \ast}X_{H})=\theta_{\varphi_{s}(p)}(X_{H})=-H(\varphi_{s}(p))=-(\varphi_{s}^{*}H)(p)
\end{equation}
and for every $Y\in\mathfrak{X}(M)$ we have
\begin{equation}
\begin{split}
(X_{H}\lrcorner \varphi_{s}^{*}d\theta)(Y) &=(\varphi_{s}^{*}d\theta)(X_{H},Y)\\
&=d\theta(\varphi_{s \ast}X_{H},\varphi_{s \ast}Y)\\
&=(X_{H}\lrcorner d\theta)(\varphi_{s \ast}Y)\\
&=\varphi_{s}^{*}(X_{H}\lrcorner d\theta)(Y)\\
&=(\varphi_{s}^{*}dH-(\varphi_{s \ast}R_{z})(\varphi_{s}^{*}H)\varphi_{s}^{*}\theta-(\varphi_{s \ast}R_{t})(\varphi_{s}^{*}H)\varphi_{s}^{*}\eta)(Y)\\
&=(d\varphi_{s}^{*}H-(\varphi_{s \ast}R_{z})(\varphi_{s}^{*}H)\varphi_{s}^{*}\theta-(R_{t}\varphi_{s}^{*}H)\eta)(Y)
\end{split}
\end{equation}
i.e.~the vector field $X_{H}$ is Hamiltonian with respect to the cocontact structure $(\varphi_{s}^{*}\theta,\eta)$ with Hamiltonian function $K_{s}=\varphi_{s}^{*}H$.

In~\cite{AE2023} it has been shown that from a given canonoid transformation for a \emph{good} cocontact Hamiltonian system one can obtain the corresponding constants of motion, provided the Hamiltonian functions are invariant under the flow of the time Reeb vector field.
It is worth remarking that in our case the cocontact Hamiltonian system needs not be good. 
Therefore in the general case we have that if $\varphi_{s}$ 
is a one-parameter group of canonoid transformations for $(M,\theta,\eta, H)$ 
with new Hamiltonian functions $K_{s}$ and if for each $s$ we have $R_{t}K_{s}=0$, then $K_{s}$ are dissipated quantities; 
in this case, the quotients $K_{s}/K_{s^{\prime}}$ (whenever well defined) are constants of motion.

\section{Conclusions}
\label{sec:conclusions}
We have considered different types of symmetries in several classes of Hamiltonian systems, 
namely symplectic, cosymplectic, contact and cocontact. In each case we
have analyzed their mutual relationships and the associated Noether-type theorems.
In our main results we have shown:  a characterization of canonoid transformations
for the cosymplectic, contact and cocontact case in a manner that is completely similar to the symplectic case (Theorems~\ref{teinfcanonoid2}, \ref{teinfcanonoid3} and~\ref{teinfcanonoid4}); with this, we have shown that scaling symmetries are always the generators of one-parameter groups of canonoid transformations (Corollaries~\ref{coroscalingcanonoid2}, \ref{coroscalingcanonoid3} and~\ref{coroscalingcanonoid4}); finally, we have discussed the conditions that guarantee that to each one-parameter group of canonoid transformations there corresponds a one-parameter family of dissipated quantities, in analogy with the symplectic case where the associated functions are conserved quantities.

Our results provide a general toolkit to study scaling symmetries and canonoid transformations in depth.
Given the current interest in these types of symmetries in the mathematical and physical literature, we expect that this work can help the qualitative analysis of several physical systems.
Moreover, we foresee that a more general definition of canonoid transformations can be applied to the cosymplectic and cocontact cases, one that does not necessarily preserve the corresponding 1-form $\eta$. It will be interesting to study this more general class of symmetries in the future. In addition, we would like to study canonoid transformations in general Jacobi structures, not necessarily the ones arising from contact or cocontact forms.
Finally, it will be also relevant to place some of the results obtained in this work within the context of the reformulation of contact systems proposed in~\cite{GG2022,grabowska2023reductions}.

\section*{Acknowledgments}
R.~Azuaje wishes to thank CONACYT (México) for the financial support through a postdoctoral fellowship in the program Estancias Posdoctorales por México 2022.
The work of A.~Bravetti was partially supported by DGAPA-UNAM, program PAPIIT, Grant No.~IA-102823.

\appendix

\section{Symplectic, cosymplectic, contact and cocontact Hamiltonian systems}
\label{sec:appA}

In this appendix we review briefly the main definitions and properties of Hamiltonian systems on symplectic, cosymplectic, contact and cocontact manifolds that are used in the main text.

\subsection{Symplectic Hamiltonian systems}
\label{subsec1}

We present a brief review of symplectic geometry and of the formulation of (time-independent) symplectic Hamiltonian mechanics (for details see \cite{AMRC2019,LR89,Torres2020,Lee2012,marsden2013introduction}).
 
\begin{de}
Let $M$ be a $2n$-dimensional smooth manifold. 
A symplectic structure on $M$ is a closed 2-form $\omega$ on $M$ such that 
$\omega^{n}\neq 0$. If $(\omega)$ is a symplectic structure on $M$, we say that $(M,\omega)$ is a symplectic manifold.
\end{de}

Let $(M,\omega)$ be a symplectic manifold of dimension $2n$. 
Locally, around any point $p\in M$, there exist local coordinates $(q^{1},\cdots,q^{n},p_{1},\cdots,p_{n})$, 
called canonical coordinates or Darboux coordinates, such that
\begin{equation}
\omega=dq^{i}\wedge dp_{i}.
\end{equation}

To each $f\in C^{\infty}(M)$ there corresponds a vector field $X_{f}$ on $M$, 
called the Hamiltonian vector field of $f$, defined through:
\begin{equation}
X_{f}\lrcorner \omega \ = \ df \ .
\end{equation}
In canonical coordinates, $X_{f}$ reads
\begin{equation}\label{eq:Xfsymplectic}
X_{f}=\frac{\partial f}{\partial p_{i}}\frac{\partial}{\partial q^{i}}-\frac{\partial f}{\partial q^{i}}\frac{\partial}{\partial p_{i}}\ .
\end{equation}
The assignment $f\longmapsto X_{f}$ is linear, that is,
\begin{equation}
X_{f+\alpha g}\ = \ X_{f}+\alpha X_{g}\ ,
\end{equation}
$\forall f,g\in C^{\infty}(M)$ and $\alpha \in\mathbb{R}$. 
Given $f,g \in C^{\infty}(M)$, the Poisson bracket of $f$ and $g$ is defined by
\begin{equation}
\lbrace f,g\rbrace=X_{g}f=\omega(X_{f},X_{g}).
\end{equation}
In canonical coordinates we have
\begin{equation}
\lbrace f,g\rbrace \ = \ \frac{\partial f}{\partial q^{i}}\frac{\partial g}{\partial p_{i}}\,-\,\frac{\partial f}{\partial p_{i}}\frac{\partial g}{\partial q^{i}} \ .
\end{equation}

The theory of time-independent (conservative) Hamiltonian systems 
is naturally constructed within the mathematical formalism of symplectic geometry. 
Given $H\in C^{\infty}(M)$, the dynamics of a system on $(M,\omega)$ with Hamiltonian function $H$ is given by the Hamiltonian vector field $X_{H}$. 
That is, the trajectories of the system 
$\psi(t)=(q^{1}(t),\cdots,q^{n}(t),p_{1}(t),\cdots,p_{n}(t))$ are the integral curves of $X_{H}$. 
From~\eqref{eq:Xfsymplectic}, we obtain that in canonical coordinates 
they satisfy Hamilton's equations of motion
\begin{equation}
\phantom{cccccccccccc}\qquad
\dot{q^{i}} =\frac{\partial H}{\partial p_{i}}, \qquad \textrm{and} 
\qquad
\dot{p_{i}} =-\frac{\partial H}{\partial q^{i}}, 
\qquad i=1,2,3,\ldots,n \ .
\end{equation}

The evolution of a function $f\in C^{\infty}(M)$ (a physical observable) 
along the trajectories of the system is given by
\begin{equation}
\dot{f}=L_{X_{H}}f=X_{H}f=\lbrace f,H\rbrace \ ,
\end{equation}
where $L_{X_{H}}f$ is the Lie derivative of $f$ with respect to $X_{H}$~\cite{AMR88}. 
We say that $f$ is a constant of motion of the system if it is constant along the trajectories of the system, 
that is, $f$ is a constant of motion if $L_{X_{H}}f=0$ (equivalently, $\lbrace f,H\rbrace=0$).

\subsection{Cosymplectic Hamiltonian systems}
\label{subsec2}

Here we provide some aspects of cosymplectic geometry and of the associated formulation of time-dependent Hamiltonian systems
(for more details see~\cite{LR89,CLL92,LS2017}). 

\begin{de}
Let $M$ be a $(2n+1)$-dimensional smooth manifold. 
A cosymplectic structure on $M$ is a couple $(\Omega,\eta)$, where $\Omega$ is a closed 2-form on $M$ and $\eta$ is a closed 1-form on $M$ 
such that $\eta\wedge\Omega^{n}\neq 0$. 
If $(\Omega,\eta)$ is a cosymplectic structure on $M$ we say that $(M,\Omega,\eta)$ is a cosymplectic manifold.
\end{de}

Let $(M,\Omega,\eta)$ be a cosymplectic manifold of dimension $2n+1$. 
Locally, around any point $p\in M$, 
there exist local coordinates $(q^{1},\cdots,q^{n},p_{1},\cdots,p_{n},t)$, called canonical coordinates or Darboux coordinates, such that
\begin{equation}
\Omega=dq^{i}\wedge dp_{i}\hspace{1cm}\textit{and}\hspace{1cm}\eta=dt.
\end{equation}
Furthermore, there exists a distinguished vector field $R$ on $M$, called the Reeb vector field, defined by
\begin{equation}
R\lrcorner \Omega =0 \hspace{1cm}\textit{and}\hspace{1cm} R\lrcorner \eta =1.
\end{equation}
In canonical coordinates we have $R=\frac{\partial}{\partial t}$.

To each $f\in C^{\infty}(M)$ there corresponds a vector field $X_{f}$ on $M$, called the Hamiltonian vector field for $f$, defined through:
\begin{equation}
X_{f}\lrcorner \Omega =df-(Rf)\eta \hspace{1cm}\textit{and}\hspace{1cm} X_{f}\lrcorner \eta =0.
\end{equation}
In canonical coordinates we have
\begin{equation}
X_{f}=\frac{\partial f}{\partial p_{i}}\frac{\partial}{\partial q^{i}}-\frac{\partial f}{\partial q^{i}}\frac{\partial}{\partial p_{i}}.
\end{equation}
The assignment $f\longmapsto X_{f}$ is linear, that is
\begin{equation}
X_{f+\alpha g}=X_{f}+\alpha X_{g},
\end{equation}
$\forall f,g\in C^{\infty}(M)$ and $\alpha \in\mathbb{R}$. 
Given $f,g \in C^{\infty}(M)$ the Poisson bracket of $f$ and $g$ is defined by
\begin{equation}
\lbrace f,g\rbrace=X_{g}f=\Omega(X_{f},X_{g}).
\end{equation}
In canonical coordinates, it reads
\begin{equation}
\lbrace f,g\rbrace=\frac{\partial f}{\partial q^{i}}\frac{\partial g}{\partial p_{i}}-\frac{\partial f}{\partial p_{i}}\frac{\partial g}{\partial q^{i}}.
\end{equation}

The theory of time-dependent Hamiltonian systems can be developed under the mathematical formalism of cosymplectic geometry 
(see \cite{LR89,CLL92,LS2017}). 
Given $H\in C^{\infty}(M)$, the dynamics of a system on $(M,\Omega,\eta)$ 
with Hamiltonian function $H$ is given by the evolution vector field $E_{H}=X_{H}+R$. 
The trajectories $\psi(t)=(q^{1}(t),\cdots,q^{n}(t),p_{1}(t),\cdots,p_{n}(t),t)$ of the system are the integral curves of $E_{H}$, 
which satisfy Hamilton's equations of motion
\begin{equation}
\dot{q^{i}} =\frac{\partial H}{\partial p_{i}}, \qquad
\dot{p_{i}} =-\frac{\partial H}{\partial q^{i}}, \qquad 
\dot t=1\,.
\end{equation}

The evolution of a function $f\in C^{\infty}(M)$ (an observable) along the trajectories of the system is given by
\begin{equation}
\dot{f}=L_{E_{H}}f=E_{H}f=X_{H}f+Rf=\lbrace f,H\rbrace+\frac{\partial f}{\partial t}.
\end{equation}
We say that a function $f\in C^{\infty}(M)$ is a constant of motion of the system if it is constant along the trajectories of the system, 
that is, $f$ is a constant of motion if $L_{E_{H}}f=0$ (equivalently, $\lbrace f,H\rbrace+\frac{\partial f}{\partial t}=0$).

\subsection{Contact Hamiltonian systems}
\label{subsec3}

In this subsection we move to the formalism of contact Hamiltonian systems, which naturally describe a large class of dissipative systems 
(for details see~\cite{LL2019,LS2017,BCT2017}), 
and have also applications in thermodynamics and statistical mechanics, among others~\cite{Bravetti2017}. 

\begin{de}
\label{deContact}
Let $M$ be a $(2n+1)$-dimensional smooth manifold. 
A contact 1-form on $M$ is a 1-form $\theta$ such that $\theta\wedge d\theta^{n}\neq 0$. 
If $\theta$ is a contact 1-form on $M$, we say that $(M,\theta)$ is a contact manifold.
\end{de}

There is a wider notion of contact manifolds: 
some authors define a contact structure on a manifold as a maximally non-integrable distribution of hyperplanes~\cite{LL2020,BH2016,Geiges2009,GG2022}. 
Locally, every contact structure is given by the kernel of a contact 1-form.
However, not every contact structure admits a global contact 1-form. 
When a contact structure admits a global contact 1-form, the contact manifold is called co-oriented. 
In this paper we restrict ourselves to co-oriented contact manifolds and we refer to them simply as contact manifolds (see~\cite{GG2022,grabowska2023reductions} for the more general perspective on contact manifolds and on the related
Hamiltonian systems).

Let $(M,\theta)$ be a contact manifold of dimension $2n+1$. 
Locally, around any point $p\in M$, there exist local coordinates $(q^{1},\cdots,q^{n},p_{1},\cdots,p_{n},z)$, called canonical or Darboux coordinates, such that
\begin{equation}
\theta=dz-p_{i}dq^{i}.
\end{equation}
Furthermore, there exists a distinguished vector field $R$ on $M$, called the Reeb vector field, defined by 
\begin{equation}
R\lrcorner \theta =1 \hspace{1cm}\textit{and}\hspace{1cm} R\lrcorner d\theta =0.
\end{equation}
In canonical coordinates we simply have $R=\frac{\partial}{\partial z}$.

To each $f\in C^{\infty}(M)$ there corresponds a vector field $X_{f}$ on $M$, called the Hamiltonian vector field of $f$, defined through:
\begin{equation}
X_{f}\lrcorner \theta = -f \hspace{1cm}\textit{and}\hspace{1cm} X_{f}\lrcorner d\theta =df-(Rf)\theta.
\end{equation}
In canonical coordinates we have
\begin{equation}
X_{f}=\frac{\partial f}{\partial p_{i}}\frac{\partial}{\partial q^{i}}-\left( \frac{\partial f}{\partial q^{i}}+p_{i}\frac{\partial f}{\partial z}\right) \frac{\partial}{\partial p_{i}}+\left( p_{i}\frac{\partial f}{\partial p_{i}}-f\right)\frac{\partial}{\partial z}.
\end{equation}
It can be checked that the assignment $f\longmapsto X_{f}$ is linear, that is
\begin{equation}
X_{f+\alpha g}=X_{f}+\alpha X_{g},
\end{equation}
$\forall f,g\in C^{\infty}(M)$ and $\alpha \in\mathbb{R}$.

While symplectic and cosymplectic manifolds are Poisson manifolds (the symplectic and cosymplectic structures define corresponding Poisson brackets), 
a contact manifold is in general a Jacobi manifold, 
i.e., the contact 1-form defines a Jacobi bracket~\cite{LL2019}. 
Given $f,g \in C^{\infty}(M)$ the Jacobi bracket of $f$ and $g$ is defined by
\begin{equation}
\lbrace f,g\rbrace=X_{g}f+fRg=\theta([X_{f},X_{g}]).
\end{equation}
In canonical coordinates we have
\begin{equation}
\lbrace f,g\rbrace=\frac{\partial f}{\partial q^{i}}\frac{\partial g}{\partial p_{i}}-\frac{\partial f}{\partial p_{i}}\frac{\partial g}{\partial q^{i}}+\frac{\partial f}{\partial z}\left( p_{i}\frac{\partial g}{\partial p_{i}}-g\right)-\frac{\partial g}{\partial z}\left( p_{i}\frac{\partial f}{\partial p_{i}}-f\right).
\end{equation}

Hamiltonian systems on contact manifolds are called contact Hamiltonian systems. 
Given $H\in C^{\infty}(M)$, the dynamics of a system on $(M,\theta)$ with Hamiltonian $H$ is defined by the Hamiltonian vector field $X_{H}$.
In canonical coordinates we have
\begin{equation}
X_{H}=\frac{\partial H}{\partial p_{i}}\frac{\partial}{\partial q^{i}}-\left( \frac{\partial H}{\partial q^{i}}+p_{i}\frac{\partial H}{\partial z}\right) \frac{\partial}{\partial p_{i}}+\left( p_{i}\frac{\partial H}{\partial p_{i}}-H\right)\frac{\partial}{\partial z}.
\end{equation}
The trajectories $\psi(t)=(q^{1}(t),\cdots,q^{n}(t),p_{1}(t),\cdots,p_{n}(t),z(t))$ of the system are the integral curves of $X_{H}$, 
which satisfy the contact version of Hamilton's equations, namely
\begin{equation}
\dot{q^{i}} =\frac{\partial H}{\partial p_{i}}, \hspace{1cm}
\dot{p_{i}} =-\frac{\partial H}{\partial q^{i}}+p_{i}\frac{\partial H}{\partial z}\,, \hspace{1cm}
\dot{z}=p_{i}\frac{\partial H}{\partial p_{i}}-H.
\end{equation}

The evolution of a function $f\in C^{\infty}(M)$ (an observable) along the trajectories of the system reads
\begin{equation}
\dot{f}=L_{X_{H}}f=X_{H}f=\lbrace f,H\rbrace-fRH.
\end{equation}
We say that a function $f\in C^{\infty}(M)$ is a constant of motion of the system if it is constant along the trajectories of the system, that is, $f$ is a constant of motion if $L_{X_{H}}f=0$ (equivalently, $\lbrace f,H\rbrace-fRH=0$).

\subsection{Cocontact Hamiltonian systems}
\label{subsec4}

This subsection, together with the previous ones, completes the notation and the language employed in this work. 
Here we briefly review the formalism of time-dependent contact Hamiltonian systems through the so-called cocontact geometry recently introduced in~\cite{Letal2022}.

\begin{de}
Let $M$ be a $(2n+2)$-dimensional smooth manifold. 
A cocontact structure on $M$ is a couple $(\theta,\eta)$ of 1-forms on $M$ such that 
$\eta$ is closed and $\eta\wedge\theta\wedge(d\theta)^{n}\neq 0$. 
If $(\theta,\eta)$ is a cocontact structure on $M$, we say that $(M,\theta,\eta)$ is a cocontact manifold.
\end{de}

Let $(M,\theta,\eta)$ be a cocontact manifold of dimension $2n+2$. 
Locally, around any point $p\in M$, 
there exist local coordinates $(t,q^{1},\cdots,q^{n},p_{1},\cdots,p_{n},z)$, called canonical coordinates or Darboux coordinates, such that
\begin{equation}
\theta=dz-p_{i}dq^{i}\hspace{1cm}\textit{and}\hspace{1cm}\eta=dt.
\end{equation}
Furthermore, there exist two distinguished vector fields $R_{z}$ and $R_{t}$ on $M$, 
called the contact Reeb vector field and the time Reeb vector field respectively, 
such that 
\begin{equation}
\left\lbrace \begin{array}{c}
R_{z}\lrcorner \eta=0 \\ 
R_{z}\lrcorner \theta=1 \\ 
R_{z}\lrcorner d\theta=0
\end{array} \right. 
\end{equation}
and
\begin{equation}
\left\lbrace \begin{array}{c}
R_{t}\lrcorner \eta=1 \\ 
R_{t}\lrcorner \theta=0 \\ 
R_{t}\lrcorner d\theta=0\,.
\end{array} \right. 
\end{equation}
In canonical coordinates, we have $R_{z}=\frac{\partial}{\partial z}$ and $R_{t}=\frac{\partial}{\partial t}$.

To each $f\in C^{\infty}(M)$ there corresponds a vector field $X_{f}$ on $M$, called the contact Hamiltonian vector field of $f$, 
defined through:
\begin{equation}
X_{f}\lrcorner \theta=-f,\hspace{1cm}X_{f}\lrcorner d\theta=df-(R_{z}f)\theta-(R_{t}f)\eta \hspace{1cm}\textit{and}\hspace{1cm} X_{f}\lrcorner \eta =0.
\end{equation}
In canonical coordinates, we have
\begin{equation}
X_{f}=\frac{\partial f}{\partial p_{i}}\frac{\partial}{\partial q^{i}}-\left( \frac{\partial f}{\partial q^{i}}+p_{i}\frac{\partial f}{\partial z}\right) \frac{\partial}{\partial p_{i}}+\left( p_{i}\frac{\partial f}{\partial p_{i}}-f\right)\frac{\partial}{\partial z}.
\end{equation}
We observe that the assignment $f\longmapsto X_{f}$ is linear, that is
\begin{equation}
X_{f+\alpha g}=X_{f}+\alpha X_{g},
\end{equation}
$\forall f,g\in C^{\infty}(M)$ and $\alpha \in\mathbb{R}$. 
Moreover, as for the case of contact manifolds, cocontact manifolds are Jacobi manifolds. 
Indeed, given $f,g \in C^{\infty}(M)$ the Jacobi bracket of $f$ and $g$ is defined by
\begin{equation}
\lbrace f,g\rbrace=X_{g}f+fR_{z}g.
\end{equation}
In canonical coordinates, we have
\begin{equation}
\lbrace f,g\rbrace=\frac{\partial f}{\partial q^{i}}\frac{\partial g}{\partial p_{i}}-\frac{\partial f}{\partial p_{i}}\frac{\partial g}{\partial q^{i}}+\frac{\partial f}{\partial z}\left( p_{i}\frac{\partial g}{\partial p_{i}}-g\right)-\frac{\partial g}{\partial z}\left( p_{i}\frac{\partial f}{\partial p_{i}}-f\right).
\end{equation}

Time-dependent contact Hamiltonian systems can be described within the framework of cocontact geometry. 
Given $H\in C^{\infty}(M)$, the dynamics of a system on $(M,\theta,\eta)$ 
with Hamiltonian $H$ is given by the evolution vector field $E_{H}=X_{H}+R_{t}$.
In canonical coordinates
\begin{equation}
E_{H}=\frac{\partial H}{\partial p_{i}}\frac{\partial}{\partial q^{i}}-\left( \frac{\partial H}{\partial q^{i}}+p_{i}\frac{\partial H}{\partial z}\right) \frac{\partial}{\partial p_{i}}+\left( p_{i}\frac{\partial H}{\partial p_{i}}-H\right)\frac{\partial}{\partial z}+\frac{\partial}{\partial t}.
\end{equation}
The trajectories $\psi(s)=(t(s),q^{1}(s),\cdots,q^{n}(s),p_{1}(s),\cdots,p_{n}(s),z(s))$ of the system are the integral curves of $E_{H}$, 
and thus they satisfy the equations
\begin{equation}
\dot{q^{i}} =\frac{\partial H}{\partial p_{i}}, \hspace{1cm}
\dot{p_{i}} =-\left( \frac{\partial H}{\partial q^{i}}+p_{i}\frac{\partial H}{\partial z}\right) , \hspace{1cm}
\dot{z}=p_{i}\frac{\partial H}{\partial p_{i}}-H, \hspace{1cm} \dot{t}=1\,.
\end{equation}
The evolution of a function $f\in C^{\infty}(M)$ (an observable) along the trajectories of the system is given by
\begin{equation}
\dot{f}=L_{E_{H}}f=E_{H}f=X_{H}f+R_{t}f=\lbrace f,H\rbrace-fR_{z}H+R_{t}f.
\end{equation}
We say that a function $f\in C^{\infty}(M)$ is a constant of motion of the system if it is constant along the trajectories of the system, that is, $f$ is a constant of motion if $L_{E_{H}}f=0$ (equivalently, $\lbrace f,H\rbrace-fR_{z}H+R_{t}f=0$).

\bibliography{refs} 
\bibliographystyle{unsrt} 

\end{document}